\documentclass[reqno]{amsart}
\usepackage{amssymb}
\usepackage[dvips]{epsfig}
\usepackage{graphicx}
\usepackage{color}
\usepackage{amsmath}
\usepackage{amssymb}
\usepackage{amsfonts}
\usepackage{amsthm}

\newcommand{\R}{\mathbb R}

\newcommand{\Z}{\mathbb Z}

\newcommand{\C}{\mathbb C}

\renewcommand{\l}{\lambda}

\newcommand{\N}{\mathbb{N}}
\newcommand{\T}{\mathbb{T}}

\newtheorem{thm}{Theorem}[section]
\newtheorem{lem}[thm]{Lemma}
\newtheorem{prop}[thm]{Proposition}
\newtheorem{cor}[thm]{Corollary}
\newtheorem{rem}{\bf Remark}[section]
\newtheorem{defn}[thm]{Definition}

\numberwithin{equation}{section}
\usepackage[colorlinks=true,pdfstartview=FitV,linkcolor=magenta,citecolor=cyan]{hyperref}
\usepackage{bm}

\begin{document}
	
\title[]{Anderson localized states for the quasi-periodic nonlinear wave equation on $\mathbb Z^d$}

\author[]{Yunfeng Shi}
\address[Y.S.]{School of Mathematics,
Sichuan University,
Chengdu 610064,
China}
\email{yunfengshi@scu.edu.cn}
\author[]{W.-M. Wang}
\address[W.W.] {CNRS and D\'epartment De Math\'ematique ,
Cergy Paris Universit\'e,
Cergy-Pontoise Cedex 95302,
France}
\email{wei-min.wang@math.cnrs.fr}



\begin{abstract} We establish large sets of Anderson localized states for the quasi-periodic nonlinear wave equation on $\mathbb Z^d$, thus extending nonlinear Anderson localization from the random  \cite{BW08} to a deterministic setting. 

\

\noindent\textsc{R\'esum\'e.} 
Nous \'etablissons de grands ensembles d'\'etats localis\'es d'Anderson pour l'\'equation d'onde non lin\'eaire quasi-p\'eriodique sur $\mathbb Z^d$, \'etendant ainsi la localisation d'Anderson non lin\'eaire du cadre al\'eatoire \cite{BW08} \`a un cadre d\'eterministe.

\

\textbf{MSC[2020]:}  37K55, 37K10, 37K45 

\

\textbf{Key words:}  Anderson localization, quasi-periodic nonlinear  wave equation, multi-scale analysis, Craig-Wayne-Bourgain method
\end{abstract}

\maketitle
	\section{Introduction and Main Result}
	Consider the nonlinear wave equation with a quasi-periodic potential: 
	\begin{align}\label{NLW}
u_{tt}+(\varepsilon \Delta+\cos(n\cdot\alpha+\theta_0)\delta_{n, n'}+m)u+\delta u^{p+1}=0, \, \text{on } \mathbb Z^d,
\end{align}
where $\varepsilon$ and $\delta$ are parameters in $[0, 1]$, $m\in [2, 3]$ and $\Delta (n, n')=\delta_{\|n-n'\|_1, 1}$  denotes the discrete Laplacian with $\|n\|_1$ the $\ell^1$ norm.
When $\delta=0$, it is well known that the linear operator 
$$H=\varepsilon \Delta+\cos(n\cdot\alpha+\theta_0)\delta_{n, n'}, \, \text{on } \mathbb Z^d,$$
exhibits Anderson localization, namely pure point spectrum with exponentially decaying eigenfunctions,
 for small $\varepsilon$ and appropriate $\alpha$ and $\theta_0$. The conditions on $\alpha$ and
 $\theta_0$ can, moreover, be made explicit, giving arithmetic Anderson localization, see e.g., \cite{CSZ23, CSZ24}, see also \cite{Jit99,JK16,JL18,GY20}.
It follows that  when $\varepsilon \ll 1$, the linear wave equation 
 $$u_{tt}+(\varepsilon \Delta+\cos(n\cdot\alpha+\theta_0)\delta_{n, n'}+m)u=0,$$
 has only Anderson localized states for all $m$, i.e., roughly speaking, wave packets localized about the origin
 remain localized for all time.  
 
The present paper addresses the persistence question when $\delta\neq 0$, namely the existence of Anderson localized type solutions for the nonlinear wave equation \eqref{NLW}
when $0<\varepsilon, \delta\ll 1$, under appropriate conditions on $\alpha$, $\theta_0$ and $m$. Since both $\varepsilon$ and $\delta$ are small, we start from the equation  
$$u_{tt}+\cos(n\cdot\alpha+\theta_0)\delta_{n, n'}u+mu=0.$$
It has solutions of the form 
$$u^{(0)}(t, x)=\sum_{n\in\Z^d} a_n\cos (\mu_nt)\delta_n(x), \, x\in\mathbb Z^d,$$
where 
$$\mu_n=\sqrt{\cos(n\cdot\alpha+\theta_0)+m},$$
and $\delta_n(x)=1$, if $x=n$ and $0$ otherwise.

We study the persistence of the above type of solutions with finite number of frequencies in time. 
Denoting the number of frequencies by $b\geq 1$,  as an Ansatz, we seek solutions to \eqref{NLW} in the form of a convergent cosine series
in time:
\begin{align}\label{Ant}
	u(t,x)=\sum_{(k, n)\in\Z^b\times\Z^d}q (k,n) \cos(k\cdot\omega t)\delta_n(x),
	\end{align} 
	or equivalently 
\begin{align*}
	u(t,n)=\sum_{k\in\Z^b}q (k,n) \cos(k\cdot\omega t),
	\end{align*}
with appropriate conditions on 	$q (k,n)$ as $|k|+|n|\to\infty$,  {where $\omega\in\R^b$ will be constructed as  the $O(\delta)$ perturbations  of  some $\mu_n$  (cf. Theorem \ref{mthm} below for details)}. Note that solutions 
of the above form is consistent with the nonlinearity in \eqref{NLW}, as the cosine series 
in  \eqref{Ant} is closed under multiplication. 

\smallskip

To state the theorem, let us introduce the Diophantine conditions on $\alpha,\theta_0$.  Fix  $\nu>0$ and define for  $N>0,$
\begin{align}
\label{alpdc}{\rm DC}_\nu&=\bigcap_{n\in\Z^d\setminus\{0\}}\left\{\alpha\in[0,1]^d:\ \min_{\xi=1,\frac12}\|{(\xi n)}\cdot\alpha\|_{\T}\geq \frac{\nu}{|n|^{2d}}\right\},\\
\label{thedc}{\rm DC}_\alpha(N)&=\bigcap_{n\in\Z^d:\ |n|\leq 2N }\left\{\theta_0\in[0,1]:\ \|\theta_0+\frac{n}{2}\cdot\alpha\|_{\T}\geq N^{-3d}\right\},
\end{align}
where $|n|=\sup_{1\leq l\leq d}|n_l|$ and $\|x\|_\T=\inf_{l\in\Z}|x-2\pi l|.$  Denote by ${\rm meas}(\cdot)$ the Lebesgue measure of a set. We have
	\begin{thm}\label{mthm}
	Fix $\alpha\in{\rm DC}_\nu, {b}\in\N$ and 
	$\theta_0\in {\rm DC}_\alpha((\varepsilon+\delta)^{-\frac{1}{10^4db^4}})$. 
	Consider on $\Z^d$ the nonlinear wave equation 
$$u_{tt}+(\varepsilon \Delta+\cos(n\cdot\alpha+\theta_0)\delta_{n,n'}+m)u+\delta u^{p+1}=0,$$
where $m\in [2, 3]$ and $p\in 2\N$.  Fix (any) $\{n^{(l)}\}_{l=1}^b\subset \Z^d$ and let $a=(a_l)_{l=1}^b\in [1,2]^b.$
Consider a solution of \eqref{NLW} when $\varepsilon=\delta=0$,
\begin{align}\label{u0}
u^{(0)}(t,n)=\sum_{l=1}^ba_l\cos(\omega_l^{(0)}t)\delta_{n^{(l)}, n},
\end{align}
where 
$$\omega^{(0)}=(\omega_l^{(0)})_{l=1}^b=\left(\sqrt{\cos(n^{(l)}\cdot\alpha+\theta_0)+m}\right)_{l=1}^b\in[1, 2]^b.$$
	For $0<\varepsilon+\delta\ll1$ and $\varepsilon=O(\delta)$,   there is $\mathcal{M}=\mathcal{M}_{\alpha,\theta_0}\subset [2,3]$ with $
	{\rm meas}([2,3]\setminus\mathcal{M})=o(1)$, such that for $m\in \mathcal{M}$, there is $\mathcal{A}=\mathcal{A}_{\alpha,\theta_0,m}\subset [1,2]^b$ with ${\rm meas}(\mathcal{A})\geq 1-(\varepsilon+\delta)^{c}$ for some {$c=c(b, d)>0$}  so that  the following holds true. If $a\in \mathcal{A}$, there is $\omega=\omega(a)$ satisfying $|\omega-\omega^{(0)}| =O(\delta)$  so that
	\begin{align*}
	u(t,n)=\sum_{k\in\Z^b}q (k,n) \cos(k\cdot\omega t)
	\end{align*}
			is a solution of \eqref{NLW}. Moreover, we have 
			\begin{align*}
			q (k,n)&=q(-k,n)\in\R,\\
			 q (\pm e_l,n^{(l)})&=a_l/2, \ 1\leq l\leq b,\\
			\sum _{(k,n)\notin \mathcal{S}} |q (k,n)|e^{|k|+|n|}&<\sqrt{\varepsilon+\delta},
			\end{align*}
			where $\mathcal{S}=\{(e_l,n^{(l)})\}_{l=1}^b\cup\{(-e_l,n^{(l)})\}_{l=1}^b$ with $e_l$ the standard basis vectors for $\Z^b.$
			\end{thm}



\subsection {The nonlinear random Schr\"odinger equation}
When the potential $V$ with $V(n)=\cos(n\cdot\alpha+\theta_0)$ is replaced
by a random potential, for example,  with $V(n)$, a family of
independently identically distributed (iid) random variables,  it is known
from {\cite{BW08} 
that the nonlinear random Schr\"odinger equation:
\begin{align*}
iu_{t}+\varepsilon\Delta +Vu+\delta |u|^{2p}u=0, \, \text{on } \mathbb Z^d,
\end{align*}
has large sets of Anderson localized states for small $\varepsilon$ and $\delta$. 
When $d=1$, it is shown further in \cite{LW24} that large sets of Anderson localized states persist for all $\varepsilon\neq 0$ 
(similar results could be proven for the 
nonlinear wave equation).  

Recall that  the proof of  Anderson localization for the linear random Schr\"odinger equation,
\begin{align*}
iu_{t}+\varepsilon\Delta +Vu=0, \, \text{on } \mathbb Z^d,
\end{align*}
for small $\varepsilon$ is  based on Green's function estimates  in \cite{FS83}, see also \cite{AM93}. 

It is well known that the linear random Schr\"odinger equation and the linear 
quasi-periodic Schr\"odinger equation share common localization features in the perturbative regime.  Theorem~\ref{mthm}
is a generalization of this to a nonlinear setting.

\subsection {Ideas of the proof and new ingredients}
Our proof  is based on  the  Craig-Wayne-Bourgain (CWB) method (cf. e.g., \cite{CW93,Bou94,Bou98,Bou05}) originated from the PDE setting. {\it  However, in the present we have to  overcome the key difficulty that the normal frequencies of linearized operators can be  dense in an interval.}

Substituting \eqref{Ant} into \eqref{NLW} leads to a nonlinear matrix equation indexed by $k$ and $n$. To order $ O(1)$, it is a diagonal 
matrix: 
\begin{equation}\label{D}
D(k, n)=\mu_n^2-(k\cdot\omega)^2=(\mu_n+k\cdot\omega)(\mu_n-k\cdot\omega).
\end{equation}
One may approximate $\omega_\ell$, 
by $\mu_{n^{(\ell)}}$, $\ell=1, 2, ..., b$. For the persistence problem, it suffices
to bound $D$ away from zero. At this level of approximation, this leads to bounding from below
the linear combinations of $\mu_n$, $n\in\mathbb Z^d$, in the monomials  in \eqref {D}. 
These linear combinations are part of the {\it harmonics} of the linear frequencies.

Bounding the harmonics from below is one of the main novelties here, and is more difficult than in the iid random case
in \cite{BW08}, and also very different from that in \cite{LW24}. There are only few parameters, namely $d+2$ parameters: $\alpha$, $\theta_0$ and $m$. The $\mu_n$'s are {\it functions} of 
these $d+2$ variables, albeit analytic. The main new difficulty here is that the harmonics $\mu_n\pm k\cdot\omega^{(0)}$ form a {\it large} family of analytic functions, which need to 
be bounded in a uniform or quantitative way.

It is useful to view the above problem as bounding from below the
inner product of two vectors $k=(k_1, k_2, ..., k_\beta)$ and $\omega=(\omega_1, \omega_2, ..., \omega_\beta)$
in $\mathbb R^\beta$, where $k\neq 0$ and $\omega=\omega(w)$ depends analytically on some parameter $w$. The idea
is then to use a generalized $\beta\times\beta$ Wronskian matrix consisting of $\beta$ derivatives of $\omega$, say the first $\beta$ derivatives,
and to show that it is invertible. If this is so, then the projection of $k$ on some $\ell$-th derivative of $\omega$, $1\leq\ell\leq\beta$ {\it must be} non-zero,
as the $\beta$ derivatives of $\omega$ span $\mathbb R^\beta$. Variation in $w$ then leads to the desired result, cf. \cite{BGG85}.  

In our case, it turns out that this Wronskian is a Vandermonde matrix and that consequently
the determinant can be bounded below by products of $\mu_n-\mu_{n'}$, $n\neq n'$. 
Providing this lower bound is the only place where we used that the potential is a cosine function.
The rest of the paper is free of this restriction. 

One may view the matrix $D$ as the initial linearized operator. To continue the analysis, one will need to study more general linearized operators
and prove their invertibility. This falls into the general category of proving large deviation theorems (LDT) for linear quasi-periodic operators. 
There are two possible approaches here. 
\vskip 0.2 cm
\noindent {\bf (i) The Diophantine frequency approach}

This consists in first showing Diophantine $\omega^{(0)}$ by using the Wronskian (Vandermonde)
determinant approach outlined above, and subsequently Diophantine $\omega$ by a simple perturbation argument, see Lemma~5.3  \cite{Wan21}.
The linear problem is then in line with using, by now, rather standard method to prove LDT for quasi-periodic operators, see e.g.,  \cite{BGS02, Bou05, HSSY,Shi22,Liu22}. 

\noindent {\bf (ii) The harmonic clusters approach}

This consists in showing that the spectrum of $D$ forms clusters of finite size at the initial stage, and as a result
the spectrum of $D$ has many gaps. The finite size clusters and the gaps can then be used to control the resonances.
The analysis at later stages is similar to (same as) (i).  Compared to (i), the measure estimates are more optimal and there is precise description
of the singular sets at the initial stages. 

For the present paper, we adopt approach (ii) for the linear analysis since we are dealing with small tangential frequencies of order ${O}(\delta)$ (rather than  order  ${O}(1)$  as in   \cite{Bou05, BW08}).  This leads to  the  {\it intermediate scales} analysis of LDT. Moreover, even at large scales  LDT, we  provide a novel  approach to establish the  {\it sublinear bound}  as compared to  \cite{Bou05} (cf. chapters 19 and 20): We employ certain  {\it  non-resonant  conditions}  by  combining   the  LDT  of quasi-periodic  Schr\"odinger operators  and  the {\it short-range property}  of the lattice equation to eliminate resonances.  
The nonlinear analysis remains, however, the same for both. For the most
updated, detailed exposition of the nonlinear analysis, see \cite{KLW24}.

Finally, we mention that the method presented has  been significantly extended to handle quasi-periodic nonlinear Schr\"odinger
equation (QPNLS) in arbitrary space  dimensions as well (cf. \cite{SW24}). Previously, Yuan \cite{Yua02} first treated  the small tangential frequencies and obtained the existence of quasi-periodic  solutions for certain nonlinear lattice equations in one dimension using  a KAM scheme  (cf. \cite{GYZ14} for  one-dimensional QPNLS case  using  the  similar KAM method).   
{\it  The existence of Anderson localized solutions for the quasi-periodic nonlinear wave equation seems, however, to have 
remained completely open until the present paper.  }

\subsection{Structure of the paper} In sect.~2, we briefly present the general scheme. Sect.~3 makes linear estimates, in the form of large deviation theorems. The paper concludes
by constructing the quasi-periodic solutions (i.e., Anderson localized states) in sect.~4. 

\section{The general scheme}

We work directly with the second order nonlinear equation \eqref{NLW}.   Let
\begin{align*}
V={\rm diag}(\mu_n)_{n\in\Z^d}\ {\rm with}\ \mu_n=\sqrt{\cos(n\cdot\alpha+\theta_0)+m}.
\end{align*}
Using the Ansatz, we obtain the time independent lattice equation 
\begin{align*}
(\mu_n^2-(k\cdot\omega)^2)q(k,n)+\varepsilon\Delta q(k,n)+\delta q_*^{p+1}(k,n)=0,
\end{align*}
where 
\begin{align}\label{involu}
q_*^{p+1}(k,n)=\sum_{k^{(1)}+\cdots+ k^{(p+1)}=k}\prod_{l=1}^{p+1}q(k^{(l)}, n),
\end{align} and the Laplacian $\Delta$ acts only on the $n$-variables. 
If we let $q=\{q(k,n)\}_{k\in\Z^b,n\in\Z^d}$, we can write the above equation as:
\begin{align}\label{nonlatt}
F(q)=Dq+\varepsilon \Delta q+\delta q_*^{p+1}=0,
\end{align}
with 
\begin{align*}
D={\rm diag}_{(k,n)\in\Z^{b+d}}(\mu_n^2-(k\cdot\omega)^2),
\end{align*}
and $q_*^{p+1}(k,n)$ given by \eqref{involu}. 
The linearized operator of $q_*^{p+1}$ at  $\tilde q$  is then
\begin{align}\label{toep}
{T}_{\tilde q}((k,n); (k',n'))=(p+1){\tilde q}_*^p(k-k',n)\delta_{n,n'}.
\end{align}
If $\tilde q(k,n)=\tilde q(-k,n)\in\R$ for all $(k,n)\in\Z^{b+d}$, then $T_{\tilde q}$ is a {\it real} self-adjoint operator on $\ell^2(\Z^{b+d})$, which is one of the reasons why we work with the cosine series. 

\subsection{The Lyapunov-Schmidt decomposition}
We look for quasi-periodic  solutions to \eqref{NLW} near $u^{(0)}$, using the Lyapunov-Schmidt decomposition.  We divide the equation \eqref{nonlatt} into the $P$-equations 
\begin{align}\label{Peq}
{F}\big |_{\mathcal{S}^c}\left(q\right)=0,\ \mathcal{S}^c=\Z^b\times\Z^d\setminus \mathcal{S}, 
\end{align}
and the $Q$-equations 
\begin{align}\label{Qeq}
{F}\big |_{\mathcal{S}}(q)=0,
\end{align}
{where $\mathcal S$ is as defined in Theorem \ref{mthm}}. 
We use \eqref{Peq} to solve for $q|_{\mathcal{S}^c}$ using a Newton scheme. On $\mathcal S$, $q$ is held fixed, {$q\big |_{\mathcal S}=q^{(0)}\big |_{\mathcal S}$ ($q^{(0)}$ is given by \eqref{q0defn})},  and the $Q$-equations
are viewed instead as equations for $\omega$, which  will be solved using the implicit function theorem.
Solving the $P$-equations requires analyzing the invertibility of the linearized operator $H=D+\varepsilon \Delta+\delta T_{ q}$ restricted to $\mathcal{S}^c$, which we do
in the next section.


\section{Linear analysis}\label{ldtsec}
	In this section we investigate  properties of linearized operators. Of particular importance is the large deviation theorem (LDT) for the corresponding Green's functions.  
	
	For a vector $v\in\R^r$, we denote $|v|$ (resp. $|v|_2$) the supremum norm (resp. the Euclidean norm). 
	Denote by $\Lambda_L=\{(k,n)\in\Z^b\times\Z^d:\ |(k,n)|\leq L\}$ the cube of radius $L$ and center $(0,0)$.  
	
	In the following we study on 
	$\Z_*^{b+d}=\Z^{b+d}\setminus \mathcal{S}$
	the operator 
	\begin{align}\label{hsgmdefn}
	{H}(\sigma)={D}(\sigma)+\varepsilon \Delta+\delta {T}_\phi,\ \sigma\in\R, 
	\end{align}
	where 
	\begin{align}
\label{dsgmdefn}{D}(\sigma)&=\mu_n^2-(\sigma+k\cdot\omega)^2,\\
\nonumber \mu_n&=\sqrt{\cos(n\cdot\alpha+\theta_0)+m}, \\
\nonumber {T}_\phi((k,n);(k',n'))&=\phi(k-k',n)\delta_{n,n'},
\end{align}
	with  $\phi:\ \Z^{b+d}_*\to\R$ satisfying $\phi(k,n)=\phi(-k,n)$ and for some $C>0, \gamma>0,$
	\begin{align}\label{gamdefn}
	|\phi(k,n)|\leq  C e^{-\gamma(|k|+|n|)}\ {\rm for}\ \forall\  (k,n)\in\Z_*^{b+d} .
	\end{align}
	   {Throughout this section, we always assume the frequency  $\omega\in\Omega$ with  \begin{align}\label{omgdefn}\Omega=\omega^{(0)}+[-C\delta, C\delta]^b\end{align}}  for some $C>0.$  
	  { \begin{rem}As we will see later in  section 4,  $T_\phi$ is typically the  linearized operator of the nonlinear $P$-equation  on some  approximate solution $\tilde q$.  We can iteratively  establish   $|\tilde q(k,n)|\leq e^{-c|k|-c|n|}$ for some $c>0$. Then we have for  any $0<c_1<c$ and   some $C=C(p,b)>1$, 
	   \begin{align*}
	   |{T}_{\tilde q}((k,n); (k',n'))|&=(p+1)|{\tilde q}_*^p(k-k',n)\delta_{n,n'}|\ ({\rm cf}. \ \eqref{involu})\\
	   &\leq C (|k-k'|+1)^{C}e^{-c|k-k'|}e^{-c|n|}\delta_{n,n'}  \\
	   &\leq C_1 e^{-c_1(|k-k'|+|n|)}\delta_{n,n'}\ {\rm for\  some}\ C_1=C_1(p,b,c_1)>0, 
		\end{align*}
	where the second inequality  follows from the  estimates  in  \cite{LW24, LSZ25} (cf. page 35 of \cite{LSZ25},  the estimate on  the  $A$ term).  
Definitely, the linearized operator is varying in a Newton scheme. This issue, however, can be
remedied by using that the Green’s function estimates in LDT  (at scale $N$) are
stable under perturbations of order $e^{-N^2}$. The sup-exponentially decaying error
 in the Newton scheme can ensure this. For more details, see section V, F. Step 3 and the proof of Proposition 3.2 in \cite{KLW24}. This is why we work  with $T_{\phi}$ in the above form. 
\end{rem}} 
		
	We introduce the elementary regions on $\Z_*^{b+d}$. 
	Denote by $\Lambda_N(x)$ ($ x\in\Z^{r},\ r=b+d,d$), $$\Lambda_N(x)=\{y\in\Z^{r}:\ | y-x|\leq N\},$$ the cube of radius $N$  centered at  $x$.	 In particular, we write  $\Lambda_N=\Lambda_N(0)$.  
An {\it elementary region}  $Q_N$ of size $N$ and centered at $ 0$ is one of the following regions 
$$Q_N=\Lambda_N\ {\rm or}\ Q_N=\Lambda_N\setminus\{x\in\Z^r:\ x_\ell \square_\ell 0,\ 1\leq\ell\leq r\},$$
 where  $\square_\ell \in \{<, >, \emptyset\}$ and at least two $\square_\ell$'s are not $\emptyset.$ Denote by $\mathcal{ER}_{0}(N)$ the set of all elementary regions of size $N$ and centered at $0$. Let $$\mathcal {ER}(N)=\{x+Q_N:\ x\in\Z^r,\ Q_N\in \mathcal{ER}_{0}(N)\}.$$

We also  define Green's functions  on subsets of $\Z_*^{b+d}$. Let $\Lambda\subset\Z_*^{b+d}$ be a non-empty set. Define (if exists)
 $$G_\Lambda(\sigma)=(R_\Lambda {H}(\sigma)R_\Lambda)^{-1},$$
	where $R_\Lambda$ denotes the restriction operator. 
	
Before addressing the LDT for Green's functions, let us first give the definition of large deviation estimates (LDE): 
\begin{defn}
Let $0<\rho_1<\rho_2<\rho_3<1$,  $0<\gamma'<\gamma$ and $M\in\N$.  We say ${H}(\sigma)$ satisfies the $(\rho_1,\rho_2,\rho_3, \gamma',M)$-LDE  if 
there exists a set $\Sigma_M\subset \R$ with $${\rm meas}(\Sigma_M)\leq e^{-M^{\rho_1}}$$ so that for $\sigma\notin \Sigma_M$ the following estimates hold: If $\Lambda\in(0,n)+ \mathcal{ER}_0(M)$ satisfies $|n|\leq 2M$, then
\begin{align*}
\|G_{\Lambda}(\sigma)\|&\leq e^{M^{\rho_2}},\\
|G_{\Lambda}(\sigma)(j;j')|&\leq e^{-\gamma'|j-j'|}\ {\rm for}\ |j-j'|\geq M^{\rho_3},
\end{align*}
{where $j=(k',n'),  j'=(k'',n'')\in {\Lambda}$}. 
\end{defn}

We will show in this section that the $(\rho_1,\rho_2,\rho_3, \gamma',N)$-LDE hold  for all $N\gg1$,
under certain non-resonant conditions on $\alpha, \theta_0, \omega$.  This leads to the LDT, Theorem~\ref{ldtthm} below. Since we treat $\omega$ as the $O(\delta)$ perturbation of $\omega^{(0)}$, the properties of $\omega^{(0)}$ become essential in the proof of LDT.  So, in the following we first establish {non-resonant properties (cf. \eqref{spb}--\eqref{sublthm}) of $\omega^{(0)}$}. Then the proof of LDT will follow from a { multi-scale analysis} scheme in the spirit of \cite{BGS02}.


	\subsection{Analysis of $\omega^{(0)}$}\label{NRomega0}
	Recall that 
	\begin{align*}
	\omega^{(0)}&=\omega^{(0)}(\alpha,\theta_0, m)=\left(\sqrt{\cos(\theta_0+n^{(l)}\cdot\alpha)+m}\right)_{l=1}^b,\\
	\mu_n&=\mu_n(\alpha,\theta_0, m)=\sqrt{\cos(\theta_0+n\cdot\alpha)+m},\ n\in\Z^d.
	\end{align*}
The main purpose of this section is to obtain  lower bounds on 
	\begin{align*}
	&|k\cdot\omega^{(0)}|\ {\rm for}\ k\in\Z^b\setminus\{0\},\\
	&|k\cdot\omega^{(0)}+\mu_n| \ {\rm for}\ (k,n)\in\Z_*^{b+d},\\
	&|k\cdot\omega^{(0)}+\mu_n-\mu_{n'}| \ {\rm for} \ n\neq n'\in\Z^d,
	\end{align*}
	assuming certain restrictions on $(\alpha, \theta_0, m).$  The general scheme is that  we first remove some $(\alpha,\theta_0)$ so that the transversality conditions on $m$ are fulfilled. Subsequently we remove $m$ so that the desired  lower bounds hold. 

	We start by imposing Diophantine conditions on  $\alpha,\theta_0$.  
	\begin{lem}\label{seplem}
	Let  $L>0$ and $c_\star\in (0,1). $ Assume that 
	\begin{align*}
	\alpha\in {\rm DC}(L;c_\star),\ \theta_0\in {\rm DC}_\alpha(L;c_\star)
	\end{align*}
	with 
	\begin{align}
	\label{dca} {\rm DC}(L;c_\star)&=\bigcap_{n:\ 0<|n|\leq 2L }\left\{\alpha\in[0,1]^d:\ \|\frac{n}{2}\cdot\alpha\|_{\T}\geq c_\star\right\},\\
	\label{dct} {\rm DC}_\alpha(L;c_\star)&=\bigcap_{n:\ |n|\leq 2L }\left\{\theta\in[0,1]:\ \|\theta+\frac{n}{2}\cdot\alpha\|_{\T}\geq c_\star\right\}.
	\end{align}	
Then we have 
\begin{align}
\label{sepeq2}&\inf_{n\neq n',\ |(n,n')|\leq L} |\mu_n^2-\mu_{n'}^2|\geq \frac{8}{\pi^2} c_\star^2,\\
\label{sepeq1}&\inf_{n\neq n',\ |(n,n')|\leq L} |\mu_n-\mu_{n'}|\geq \frac{2}{\pi^2} c_\star^2.
\end{align}
	\end{lem} 
	\begin{proof}
It suffices to note that 
\begin{align*}
|\mu_n^2-\mu_{n'}^2|&=|2\sin\frac{(n-n')\cdot\alpha}{2}\cdot\sin\frac{2\theta_0+(n+n')\cdot\alpha}{2}|\\
&\geq \frac{8}{\pi^2}\|\frac{(n-n')\cdot\alpha}{2}\|_{\T}\cdot\|\frac{2\theta_0+(n+n')\cdot\alpha}{2}\|_{\T}
\end{align*}
and 
\begin{align*}
|\mu_n^2-\mu_{n'}^2|&\leq4 |\mu_n-\mu_{n'}|.
\end{align*}
\end{proof}
Next, we verify the transversality conditions.  
 For convenience, define for any $\beta\in\N$ and any $n^{(l)}\in\Z^d$ ($1\leq l\leq \beta$), 
$$\tilde\omega=\left(\sqrt{\cos(n^{(l)}\cdot\alpha+\theta_0)+m}\right)_{l=1}^\beta\in[1, 2]^\beta.$$
We have
\begin{prop}\label{tvprop}
 Let $L>\sup\limits_{1\leq l\leq \beta}|n^{(l)}|$ and $ c_\star\in (0,1)$. Assume $\alpha\in{\rm DC}{(L;c_\star)}$ and $\theta_0\in {\rm DC}_\alpha{(L;c_\star)}$. Then there is $\tilde c=\tilde c(\beta)>0$ depending only on $\beta$ so that 
for $k\in\Z^\beta\setminus\{0\}$, 
\begin{align}\label{tdc}
\inf_{m\in[2, 3]}\sup_{1\leq l\leq \beta}\left|\frac {\partial^lk\cdot\tilde\omega}{\partial m^l}(m)\right| \geq \tilde c c_\star^{\beta(\beta-1)}|k|_2.
\end{align}

\end{prop}

The following lemma is useful to prove \eqref{tdc}.
\begin{lem}[Proposition in Appendix B, \cite{BGG85}]\label{BGGlem}
Let $v^{(1)},\cdots, v^{(r)}\in\R^r$ be $r$ linearly independent vectors with $|v^{(l)}|_1=\sum\limits_{1\leq j\leq r}|v_j^{(l)}|\leq K$ for $1\leq l\leq r.$ Let $w\in\R^r$ satisfy 
\begin{align*}
\sup_{1\leq l\leq r}|w\cdot v^{(l)}|<\xi\ {\rm for\  some}\ \xi>0.
\end{align*}
Then one has 
\begin{align*}
|w|_2\cdot|\det \left(v_j^{(l)}\right)_{1\leq l,j\leq r}|<rK^{r-1}\xi.
\end{align*}
\end{lem}

\begin{proof} [Proof of Proposition~\ref{tvprop}]
Assume $v(m)=\tilde\omega(m)$ is the row vector.  Consider the $\beta\times \beta$ matrix  ${\bf M}(m)$ with the $l$-th row given by 
$${\bf M}_l(m)=\left(\frac{\partial^lv_s(m)}{\partial m^l}\right)_{s=1}^{\beta}.$$
We want to estimate the Wronskian $\det {\bf M}$. A key observation is that 
\begin{align*}
\frac{\partial^lv_s(m)}{\partial m^l}=\frac12(\frac12-1)\cdots(\frac12-l+1)v_s^{-(2l-1)}(m):=\lambda_lv_s^{-(2l-1)}(m).
\end{align*}
Then 
\begin{align*}
|\det {\bf M}(m)|=\prod_{s=1}^\beta|\lambda_sv_s^{-1}(m)|^\beta\cdot |\det {\bf M}'(m)|,
\end{align*}
where ${\bf M}'(m)$ is a Vandermonde matrix with the $l$-th  ($1\leq l\leq \beta$) row given by
$${\bf M}_l'(m)=\left(v_s^{-2(l-1)}(m)\right)_{s=1}^\beta.$$ 
 So we obtain 
 \begin{align}\label{detm1}
|\det {\bf M}(m)|=\prod_{s=1}^\beta|\lambda_sv_s^{-1}(m)|^\beta \prod_{1\leq l_1<l_2\leq \beta}|v_{l_1}^{-2}(m)-v_{l_2}^{-2}(m)|.
\end{align}
 Recall that $v_s(m)\in[1,2]$, $\sup_{1\leq l\leq \beta}|n^{(l)}|<L$.   Thus  by Lemma \ref{seplem},  we obtain
 \begin{align*}
|v_{l_1}^{-2}(m)-v_{l_2}^{-2}(m)|&=|v_{l_1}^{-2}(m)v_{l_2}^{-2}(m)|\cdot|\mu_{n^{(l_1)}}^2-\mu_{n^{(l_2)}}^2|\\
&\geq\frac{1}{2\pi^2}c_\star^2,
 \end{align*}
 which combined with \eqref{detm1} implies 
 \begin{align}\label{detm}
|\det {\bf M}(m)|\geq c'c_\star^{\beta(\beta-1)},
 \end{align}
where $c'=c'(\beta)>0$ depends only on $\beta$.

Now, we fix $k\in\Z^\beta\setminus\{0\}$ and $m\in[2,3]$. Then $v_s(m)\in[1, 2]$. Observe first that for $1\leq l\leq \beta$, 
\begin{align*}
\sup_{1\leq l\leq \beta}|{\bf M}_l(m)|_1&\leq \sup_{1\leq l\leq \beta}\sum_{s=1}^\beta|\lambda_lv_s^{-l}|\leq \sup_{1\leq l\leq \beta}\sum_{s=1}^\beta|\lambda_l|:=K_1(\beta).
\end{align*} 
Applying Lemma~\ref{BGGlem} with $r=\beta, v^{(l)}={\bf M}_l(m)$, $K=K_1(\beta), w=k$ and choosing 
\begin{align*}
\xi=\frac12 \beta^{-1} K_1(\beta)^{1-\beta}c'c_\star^{\beta(\beta-1)}|k|_2,
\end{align*}
we obtain using \eqref{detm}
\begin{align*}
\beta K_1(\beta)^{\beta-1}\xi&=\frac12 c'c_\star^{\beta(\beta-1)}|k|_2\\
&\leq \frac 12 |k|_2\cdot|\det ({\bf M}_l(m))_{1\leq l\leq \beta}|.
\end{align*}
Using Lemma~\ref{BGGlem} yields 
\begin{align*}
\sup_{1\leq l\leq \beta}\left|\frac{\partial ^lk\cdot\tilde\omega}{\partial m^l}(m)\right|&=\sup_{1\leq l\leq \beta}|k\cdot {\bf M}_l(m)|\\
&\geq \xi=\tilde c |k|_2c_\star^{\beta(\beta-1)},
\end{align*}
where $\tilde c=\tilde c(\beta)>0$ depends only on $\beta.$
 This proves \eqref{tdc}. 
 \end{proof}
 
 Next, we apply the above proposition with $\beta=b, b+1$ and $b+2$ to establish

\begin{cor}\label{tvpcor}
Let $L>\sup\limits_{1\leq l\leq b}|n^{(l)}|$ and $ c_\star\in (0,1)$. Assume $\alpha\in{\rm DC}{(L;c_\star)}$ and $\theta_0\in {\rm DC}_\alpha{(L;c_\star)}$. 
 Then there is $\tilde c=\tilde c(b)>0$ depending  only on $b$ so that
\begin{itemize}
\item[(1)] 
For $k\in\Z^b\setminus\{0\}$, 
\begin{align*}
\inf_{m\in[2, 3]}\sup_{1\leq l\leq b}\left|\frac {\partial^lk\cdot\omega^{(0)}}{\partial m^l}(m)\right| \geq \tilde c c_\star^{b(b-1)}|k|_2.
\end{align*}

\item[(2)]Fix $(k,n)\in \Z^{b+d}\setminus\mathcal{S}$ so that $|n|\leq L$. If $n\notin \{n^{(l)}\}_{1\leq l\leq b}$, then 
\begin{align}\label{tmk11}
\inf_{m\in[2, 3]}\sup_{1\leq l\leq b+1}\left|\frac {\partial^l(k\cdot\omega^{(0)}+\mu_n)}{\partial m^l}(m)\right| \geq \tilde c c_\star^{b(b+1)}|(k,1)|_2.
\end{align}
If $n=n^{(l')}$ for some $1\leq l'\leq b$, then 
\begin{align}\label{tmk12}
\inf_{m\in[2, 3]}\sup_{1\leq l\leq b}\left|\frac {\partial^l(k\cdot\omega^{(0)}+\mu_n)}{\partial m^l}(m)\right| \geq \tilde c c_\star^{b(b-1)}|k+e_{l'}|_2.
\end{align}

\item[(3)] 
Fix $n\neq n'\in\Z^d$ with $|(n, n')|\leq L$.  If 
$$(n,n')\notin \{n^{(l)}\}_{1\leq l\leq b}\times\{n^{(l)}\}_{1\leq l\leq b},$$
then for all $k\in\Z^b$, 
\begin{align}\label{tmk21}
\inf_{m\in[2, 3]}\sup_{1\leq l\leq b+2}\left|\frac {\partial^l(k\cdot\omega^{(0)}+\mu_n-\mu_{n'})}{\partial m^l}(m)\right| \geq \tilde c c_\star^{(b+1)(b+2)}C_{k,n,n'},
\end{align}
where  
\begin{align*}
C_{k,n,n'}&=\left\{
	\begin{array}{ll}
|(k,1,-1)|_2& {\rm if} \ n,n'\notin \{n^{(l)}\}_{1\leq l\leq b},\\
|(k+e_{l'}, -1)|_2& {\rm if} \  n=n^{(l')}, \ n'\notin \{n^{(l)}\}_{1\leq l\leq b},\ 1\leq l'\leq b, \\
 |(k-e_{l'}, 1)|_2&{\rm if}\  n\notin \{n^{(l)}\}_{1\leq l\leq b},\ n'=n^{(l')},\  1\leq l'\leq b.
\end{array}\right.
\end{align*}
If $n=n^{(l')}$ and $n'=n^{(l'')}$ for some $1\leq l'\neq l''\leq b,$ then for $k\neq -e_{l'}+e_{l''}$,
\begin{align}\label{tmk22}
\inf_{m\in[2, 3]}\sup_{1\leq l\leq b}\left|\frac {\partial^l(k\cdot\omega^{(0)}+\mu_n-\mu_{n'})}{\partial m^l}(m)\right| \geq \tilde c c_\star^{b(b-1)}|k+e_{l'}-e_{l''}|_2.
\end{align}
\end{itemize}
\end{cor}

 \begin{proof}

 \begin{itemize}
\item[(1)] It suffices to let $\beta=b$ and $\tilde\omega=\omega^{(0)}$ in Proposition~\ref{tvprop}.
\item[(2)] 
 If   $n\not\in \{n^{(l)}\}_{1\leq l\leq b}$,  it suffices to consider derivatives estimates of  
$
\tilde k\cdot v(m),
$
where
 $$\tilde k=(k,1)\in\Z^{b+1}\setminus\{0\},\ v(m)=(\omega^{(0)}, \mu_n)\in\R^{b+1}.$$
Similar to the proof of \eqref{tdc}, we  have 
\begin{align*}
\inf_{m\in[2, 3]}\sup_{1\leq l\leq b+1}\left|\frac {\partial^l(k\cdot\omega^{(0)}+\mu_n)}{\partial m^l}(m)\right| \geq c(b) c_\star^{b(b+1)}|(k,1)|_2,
\end{align*}
which implies \eqref{tmk11}.

If   $n=n^{(l')}$ for some $1\leq l'\leq b,$  we have since  $(k,n)\in \Z^{b+d}\setminus\mathcal{S}$ that $k+e_{l'}\in\Z^b\setminus\{0\}.$ Then we get 
\begin{align*}
k\cdot\omega^{(0)}+\mu_n=(k+e_{l'})\cdot\omega^{(0)},
\end{align*}
which combined with similar arguments as above shows
\begin{align*}
\inf_{m\in[2, 3]}\sup_{1\leq l\leq b}\left|\frac {\partial^l(k\cdot\omega^{(0)}+\mu_n)}{\partial m^l}(m)\right| \geq c(b) c_\star^{b(b-1)}|k+e_{l'}|_2.
\end{align*}
This implies \eqref{tmk12}.
\item[(3)]  
We have the following cases. \\
{\bf Case 1.} Both $n\notin\{n^{(l)}\}_{1\leq l\leq b}$ and $n'\notin\{n^{(l)}\}_{1\leq l\leq b}$. In this case it suffices to estimate derivatives of $\tilde k\cdot v(m)$ with
\begin{align*}
\tilde k=(k,1,-1)\in\Z^{b+2}\setminus\{0\},\ v(m)=(\omega^{0},\mu_n,\mu_{n'})\in\R^{b+2}.
\end{align*}
So similar to the proof of \eqref{tdc},  we have 
\begin{align*}
\inf_{m\in[2, 3]}\sup_{1\leq l\leq b+2}\left|\frac {\partial^l(k\cdot\omega^{(0)}+\mu_n-\mu_{n'})}{\partial m^l}(m)\right| \geq c(b) c_\star^{(b+1)(b+2)}|(k,1,-1)|_2,
\end{align*}
which shows \eqref{tmk21}. \\
{\bf Case 2.} Assume that $n\notin\{n^{(l)}\}_{1\leq l\leq b}$ and $n'\in\{n^{(l)}\}_{1\leq l\leq b}$ (the proof of the case  $n\in\{n^{(l)}\}_{1\leq l\leq b}$ and $n'\notin\{n^{(l)}\}_{1\leq l\leq b}$ is similar).  In this case assume $n'=n^{(l')}$ for some $1\leq l'\leq b$. It suffices to estimate derivatives of $\tilde k\cdot v(m)$ with 
\begin{align*}
\tilde k=(k-e_{l'},1)\in\Z^{b+1}\setminus\{0\},\ v(m)=(\omega^{(0)},\mu_n)\in\R^{b+1}
\end{align*}
since 
$$k\cdot \omega^{(0)}+\mu_{n}-\mu_{n'}=(k-e_{l'})\cdot\omega^{(0)}+\mu_n.$$
Thus we obtain 
\begin{align*}
\inf_{m\in[2, 3]}\sup_{1\leq l\leq b+1}\left|\frac {\partial^l(k\cdot\omega^{(0)}+\mu_n-\mu_{n'})}{\partial m^l}(m)\right| \geq c(b) c_\star^{b(b+1)}|(k-e_{l'},1)|_2,
\end{align*}
which shows \eqref{tmk21}.\\
{\bf Case 3.}  Assume $n=n^{(l')}, n'=n^{(l'')}$ for some $1\leq l'\neq l''\leq b.$ In this case we aim to estimate derivatives of $\tilde k\cdot v(m)$ with
\begin{align*}
\tilde k=(k+e_{l'}-e_{l''})\in\Z^{b}\setminus\{0\},\ v(m)=\omega^{(0)}\in\R^{b}
\end{align*}
since 
$$k\cdot \omega^{(0)}+\mu_{n}-\mu_{n'}=(k+e_{l'}-e_{l''})\cdot\omega^{(0)}.$$
Note that $k+e_{l'}-e_{l''}\neq 0$. We have by applying \eqref{tdc} that 
\begin{align*}
\inf_{m\in[2, 3]}\sup_{1\leq l\leq b}\left|\frac {\partial^l(k\cdot\omega^{(0)}+\mu_n-\mu_{n'})}{\partial m^l}(m)\right| \geq c(b) c_\star^{b(b-1)}|k+e_{l'}-e_{l''}|_2,
\end{align*}
which yields \eqref{tmk22}.
\end{itemize}

\end{proof}

Based on the above results, we can make desired restrictions on $m\in[2, 3]$ so that the {non-resonant conditions (cf. \eqref{spb}--\eqref{sublthm})} hold.    Let $\mathcal{J}_1$ be the set of all $(k,n,n')\in\Z^{b}\times\Z^d\times\Z^d$ satisfying 
\begin{align*}
 |k|\leq 2L,\  n\neq n', \ |(n,n')|\leq L,\  (n,n')\notin \{n^{(l)}\}_{l=1}^b\times \{n^{(l)}\}_{l=1}^b.
\end{align*}
The  $\mathcal{J}_2$ is defined to be the set of  $(k,n,n')\in\Z^{b}\times\Z^d\times\Z^d$ with
\begin{align*}
 |k|\leq 2L,\ \exists\  1\leq l'\neq l''\leq b,\ {\rm s.t.,}\  n=n^{(l')}, \  n'=n^{(l'')}, \ k+e_{l'}-e_{l''}\neq 0.
\end{align*}
Under the assumptions of Corollary~\ref{tvpcor}, we can define 
\begin{defn}
Fix $\eta\in (0,1)$. Let $\mathcal{M}^{(l)}\subset[2, 3]$ ($l=0,1, 2$) be defined as follows:
\begin{align*}
&\mathcal{M}^{(0)}=\bigcap_{k\in\Z^b,\ 0<|k|\leq 2L}\left\{m\in[2,3]:\  |k\cdot\omega^{(0)}|>\eta\right\},\\
&\mathcal{M}^{(1)}=\bigcap_{(k,n)\in\Lambda_L\setminus\mathcal{S}}\left\{m\in[2,3]:\  |k\cdot\omega^{(0)}+\mu_n|>\eta\right\},\\
&\mathcal{M}^{(2)}=\bigcap_{(k,n,n')\in (\mathcal{J}_1\cup\mathcal{J}_2) }\left\{m\in[2,3]:\  |k\cdot\omega^{(0)}+\mu_n-\mu_{n'}|>\eta\right\}.
\end{align*}

\end{defn}
In the following we estimate the measure  of $\mathcal{M}^{(l)}$ ($l=0,1,2$). For this purpose, one uses transversality estimates provided by  Corollary \ref{tvpcor}. 
Similar estimates have been considered previously by \cite{Kle05}  in the context of Anderson localization for quasi-periodic Schr\"odinger operators with Gevrey regular potentials. The main issue is to derive  {a}  quantitative (upper) bound on
\begin{align}\label{mesf}
{\rm meas}(\{x\in I:\ |f(x)|\leq \eta\}),
\end{align}
where $I\subset \R$ is a finite interval and $f$ is a smooth function satisfying certain transversality conditions.  
The result of \cite{Kle05}  (cf. Lemma 5.3) can not be used directly here since  the estimates there  depend implicitly on  $f$.  However, our functions $f$ rely on the lattice sites $k\in\Z^b, n\in\Z^d$. So  it requires the estimates in  \eqref{mesf} 
to depend explicitly on information on  $f$.  This will be given by 
\begin{lem}\label{lajlem}
Fix $r\in \N$. Let $I\subset\R$ be a closed interval of finite length (denoted by $|I|$). Assume that  the function $f\in C^{r+1}(I;\R)$ satisfies for some $\tau\in(0,1),$
\begin{align}\label{laj1}
\inf_{x\in I}\sup_{1\leq l\leq r}\left|\frac{d^lf}{d x^l}(x)\right|\geq \tau. 
\end{align}
Let
\begin{align}\label{laj2}
A=\sup_{x\in I}\sup_{1\leq l\leq r+1}\left|\frac{d^lf}{d x^l}(x)\right|<\infty. 
\end{align}
Then for $\eta\in(0,1),$
\begin{align}\label{laj3}
{\rm meas}(\{x\in I:\ |f(x)|\leq \eta\})\leq  \frac{CA|I|}{\tau^2}\eta^{\frac1r},
\end{align}
where $C=C(r)>0$ depends only on $r$, but not on $f.$
\end{lem}

To prove \eqref{laj3}, we first recall a  useful lemma. 
\begin{lem}[cf. e.g., Lemma A.5 of \cite{SW24}]\label{xyq}
Let $I\subset\R$ be a closed interval of finite length and  let $f\in C^r(I;\R)$. Assume that 
\begin{align}\label{xyq1}
\inf_{x\in I}\left|\frac{d^rf}{d x^r}(x)\right|\geq \frac\tau2. 
\end{align}
Then 
\begin{align}\label{xyq2}
{\rm meas}(\{x\in I:\ |f(x)|\leq \eta\})\leq C(r)(2\eta\tau^{-1})^{\frac1r},
\end{align}
where $C(r)=r(r+1)((r+1)!)^{\frac1r}.$
\end{lem}

\begin{proof}[Proof of Lemma~\ref{lajlem}]We divide $I$ into small subintervals and verify the condition \eqref{xyq1} of Lemma \ref{xyq} on each subinterval.  Let 
$$N=\left[ \frac{2A|I|}{\tau}\right]+1\in\N,$$ 
where $[x]$ denotes the integer part of $x\in\R.$ Define subintervals of $I=[h_0, h_1]$ to be 
$$I_i=\left[h_0+(i-1)\frac{|I|}{N}, h_0+i\frac{|I|}{N}\right],\ i\in[1,N].$$
Then we have 
\begin{align}\label{plaj1}
\bigcup_{i=1}^N I_i=I,\ \sup_{1\leq i\leq N}|I_i|\leq \frac{\tau}{2A}.
\end{align}
Now fix  $1\leq i\leq N$ and $x_0\in I_i$. From \eqref{laj1},  there is some $r_1\in [1,r]$ so that 
\begin{align}\label{plaj2}
\left|\frac{d^{r_1}f}{d x^{r_1}}(x_0)\right|\geq \tau>0. 
\end{align}
Combining  \eqref{laj2}, \eqref{plaj1} and \eqref{plaj2} shows
\begin{align*}
\inf_{x\in I_i}\left|\frac{d^{r_1}f}{d x^{r_1}}(x)\right|&\geq\left|\frac{d^{r_1}f}{d x^{r_1}}(x_0)\right| -\sup_{x\in I_i}\left|\frac{d^{r_1}f}{d x^{r_1}}(x)-\frac{d^{r_1}f}{d x^{r_1}}(x_0)\right|\\
&\geq \tau-\sup_{x\in I_i}\left|\frac{d^{r_1+1}f}{d x^{r_1+1}}(x)\right|\cdot |I_i|\\
&\geq \tau-A\cdot\frac{\tau}{2A}=\frac\tau2.
\end{align*}
Applying Lemma \ref{xyq} and noting $\tau,\eta\in (0,1)$ and $r_1\leq r$, we have 
\begin{align*}
{\rm meas}(\{x\in I_i:\ |f(x)|\leq \eta\})&\leq \frac{C(r)}{\tau}\eta^{\frac 1r}.
\end{align*} 
As a result,
\begin{align*}
{\rm meas}(\{x\in I:\ |f(x)|\leq \eta\})&\leq \sum_{i=1}^N\frac{C(r)}{\tau}\eta^{\frac 1r}\leq \frac{C(r)A|I|}{\tau^2}\eta^{\frac 1r},
\end{align*} 
where $C(r)>0$ depends only on $r.$ This proves \eqref{laj3}. 

\end{proof}
We are ready to prove the measure bounds.  
\begin{prop}\label{misize}
Fix $\eta\in (0,1)$. Under the assumptions of  Corollary \ref{tvpcor}, we have 
\begin{align}
\label{m0size}{\rm meas}([2,3]\setminus \mathcal{M}^{(0)})&\leq C(b)L^bc_\star^{-2b(b-1)}\eta^{\frac1b},\\
\label{m1size}{\rm meas}([2,3]\setminus \mathcal{M}^{(1)})&\leq C(b,d)L^{b+d}c_\star^{-2b(b+1)}\eta^{\frac{1}{b+1}},\\
\label{m2size}{\rm meas}([2,3]\setminus \mathcal{M}^{(2)})&\leq C(b,d)L^{b+2d}c_\star^{-2(b+1)(b+2)}\eta^{\frac{1}{b+2}}.
\end{align}
\end{prop}

\begin{proof}
We will apply Lemma~\ref{lajlem} with 
\begin{align*}
f(m)=k\cdot\omega^{(0)}(m), \ k\cdot\omega^{(0)}(m)+\mu_n(m),\ k\cdot\omega^{(0)}(m)+\mu_n(m)-\mu_{n'}(m)
\end{align*}
and $I=[2,3]$. 
We consider $f=k\cdot\omega^{(0)}$. The condition \eqref{laj1} is satisfied with 
$$\tau=\tilde c c_\star^{b(b-1)}|k|_2,\  r=b(b-1).$$
It remains to verify \eqref{laj2}. In fact, direct computation yields 
\begin{align*}
A=\sup_{m\in[2,3]}\sup_{1\leq l\leq b+1}\left|\frac{d^lf}{d m^l}(m)\right|\leq C(b)|k|,
\end{align*}
where $C(b)>0$ only depends on $b$. Then applying Lemma \ref{lajlem} implies 
\begin{align*}
{\rm meas}([2,3]\setminus \mathcal{M}^{(0)})&\leq C(b)(4L+1)^bc_\star^{-2b(b-1)}\eta^{\frac 1b}\\
&\leq C(b)L^bc_\star^{-2b(b-1)}\eta^{\frac1b}.
\end{align*}
Similarly, we have 
\begin{align*}
{\rm meas}([2,3]\setminus \mathcal{M}^{(1)})&\leq C(b,d)L^{b+d}c_\star^{-2b(b+1)}\eta^{\frac{1}{b+1}},\\
{\rm meas}([2,3]\setminus \mathcal{M}^{(2)})&\leq C(b,d)L^{b+2d}c_\star^{-2(b+1)(b+2)}\eta^{\frac{1}{b+2}}.
\end{align*}
This finishes the proof.
\end{proof}
Finally, we introduce our main theorem in this section.  

\begin{thm}\label{clusthm}
There is $L_0=L_0(b,d)>1$  with the following properties. Assume that  
$L\geq L_0\sup\limits_{1\leq l\leq b}|n^{(l)}|$,
 and   $\eta>0$ satisfies
\begin{align}\label{DD}
L^{50db^2}\eta^{\frac{1}{b+2}}<1.
\end{align}
Assume further that  $\alpha\in{\rm DC}(L;L^{-3d}) \ ({\rm cf.}\ \eqref{dca}) $ and $\theta_0\in {\rm DC}_\alpha(L;L^{-3d})\ ({\rm cf.}\ \eqref{dct}) $. Then there exists  a set $\mathcal{M}=\mathcal{M}_{\alpha,\theta_0, L,\eta}\subset [2, 3]$ with 
\begin{align}\label{msizethm}
{\rm meas}([2,3]\setminus \mathcal{M})&\leq L^{50db^2}\eta^{\frac{1}{b+2}},
\end{align}
so that  the following statements hold for $m\in\mathcal{M}$:  
\begin{itemize}
\item[(1)]  For all $n\neq n'\in \Z^d$ with $|(n,n')|\leq L$, we have 
\begin{align}\label{spb}
|\mu_n-\mu_{n'}|\geq \frac{2}{\pi^2} L^{-6d}.
\end{align}
\item[(2)] For all $k\in\Z^{b}\setminus \{0\}$ with $|k|\leq 2L$, we have
\begin{align}\label{dcthm}
 |k\cdot\omega^{(0)}|>\eta.
\end{align}
\item [(3)] For all $(k,n)\in\Lambda_L\setminus\mathcal{S}$,  we have 
\begin{align}\label{mk1thm}
 |k\cdot\omega^{(0)}+\mu_n|>\eta.
\end{align}
\item[(4)]  We have 
\begin{align}\label{sublthm}
\sup_{\xi=\pm1,\ \sigma\in\R}\#\left\{(k,n)\in \Lambda_L:\  |\xi(\sigma+k\cdot\omega^{(0)})+\mu_n|<\frac\eta2\right\}\leq b,
\end{align}
where $\#(\cdot)$ denotes the cardinality of a set.
\end{itemize}

\end{thm}
 \begin{rem} The estimates in \eqref{dcthm} can be written for all $k$ and $\eta$ satisfying 
 the constraint \eqref{DD}. This leads to Diophantine $\omega^{(0)}$.  
 We also note that the estimate \eqref{mk1thm} is used in the nonlinear analysis only. 
 \end{rem}

\begin{proof}
The proofs of (1)--(3)  are  direct corollaries  of  Lemma~\ref{seplem}  and Proposition~\ref{misize}. In fact, since $c_\star=L^{-3d}$,    $\alpha\in{\rm DC}(L;L^{-3d})$ and $\theta_0\in {\rm DC}_\alpha(L;L^{-3d})$, 
the spacing lower bound \eqref{spb} follows from \eqref{sepeq1} of Lemma~\ref{seplem}. 
Next, it suffices to let 
$$\mathcal{M}=\bigcap_{0\leq l\leq 2}\mathcal{M}^{(l)},$$
where  each $\mathcal{M}^{(l)}$ is give by Proposition \ref{misize}.  
For the proof of measure bound \eqref{msizethm},  using \eqref{m0size}--\eqref{m2size}  and $c_\star=L^{-3d}$,  we have  that
\begin{align*}
{\rm meas}([2,3]\setminus \mathcal{M})&\leq C(b,d)L^{b+2d}L^{6d(b+1)(b+2)}\eta^{\frac{1}{b+2}}\\
&\leq C(b,d)L^{3db+36db^2}\eta^{\frac{1}{b+2}}\\
&\leq L^{50db^2}\eta^{\frac{1}{b+2}},
\end{align*}
where in the last inequality we used $L\geq L_0(b,d)\gg1.$ Assume $m\in\mathcal{M}$. Then \eqref{dcthm} and \eqref{mk1thm}  follow from Proposition~\ref{misize}. 

It remains to establish \eqref{sublthm} of (4).  We prove it by contradiction.  
It suffices to consider the case $\xi=1$.  Assume  that there are $b+1$ distinct   $\{(k^{(l)}, n_*^{(l)})\}_{1\leq l\leq b+1}$ satisfying 
\begin{align}\label{sings}
|\sigma+k^{(l)}\cdot\omega^{(0)}+\mu_{n_*^{(l)}}|<\frac\eta2.
\end{align}
First, we claim that all  $k^{(l)}$ ($1\leq l\leq b+1$) are distinct. Actually, if there are $k^{(l_1)}=k^{(l_2)}$ for some $1\leq l_1\neq l_2\leq b+1$, then it must be that $n_*^{(l_1)}\neq n_*^{(l_2)}$. As a result, we obtain using \eqref{sings} that
\begin{align*}
|\mu_{n_*^{(l_1)}}-\mu_{n_*^{(l_2)}}|\leq |\sigma+k^{(l_1)}\cdot\omega^{(0)}+\mu_{n_*^{(l_1)}}|+|\sigma+k^{(l_2)}\cdot\omega^{(0)}+\mu_{n_*^{(l_2)}}|\leq \eta,
\end{align*}
which contradicts \eqref{spb}. Next, we claim that  $\{n_*^{(l)}\}_{1\leq l\leq  b+1}\subset \{n^{(l)}\}_{1\leq l\leq b}.$ In fact, if $n_*^{(l_1)}\notin  \{n^{(l)}\}_{1\leq l\leq b}$ for some $1\leq l_1\leq b+1$, we can assume that there is $l_2\neq l_1$ so that $n_*^{(l_2)}\neq n_*^{(l_1)}$. Otherwise, we may have 
\begin{align*}
|(k^{(l_1)}-k^{(l_2)})\cdot\omega^{(0)}|\leq \eta,
\end{align*}
which contradicts \eqref{dcthm}, since $k^{(l_1)}\neq k^{(l_2)}$ as claimed above. 
So  we obtain for some  $l_2\neq l_1$ and $1\leq l_2\leq b=1$  with $n_*^{(l_2)}\neq n_*^{(l_1)}$ that
\begin{align*}
|(k^{(l_1)}-k^{(l_2)})\cdot\omega^{(0)}+\mu_{n_*^{(l_1)}}-\mu_{n_*^{(l_2)}}|\leq\eta,
\end{align*}
which contradicts the definition of $\mathcal{M}^{(2)}$ and  $m\in \mathcal{M}.$ So  it suffices to assume $\{n_*^{(l)}\}_{1\leq l\leq b+1}\subset \{n^{(l)}\}_{1\leq l\leq b}.$ However, in this case we have  by  the pigeonhole principle  that there are $1\leq l_1\neq l_2\leq b+1$ so that $n_*^{(l_1)}=n_*^{(l_2)}$. Then we get 
\begin{align*}
|(k^{(l_1)}-k^{(l_2)})\cdot \omega^{(0)}|\leq \eta,
\end{align*}
which contradicts \eqref{dcthm} as shown above.  This proves \eqref{sublthm}. 

\end{proof}
	
		\subsection{Large deviation theorem}
	Throughout this section we will choose 
	\begin{align}\label{etaed}
 \eta=(\varepsilon+\delta)^{\frac{1}{8b}}, \ L=(\varepsilon+\delta)^{-\frac{1}{10^4db^4}}
	\end{align}
	in Theorem~\ref{clusthm}. 
	So, the set $\mathcal{M}\subset [2, 3]$ given in Theorem~\ref{clusthm}  will satisfy 
	$$	{\rm meas}([2,3]\setminus\mathcal{M})\leq L^{50db^2}\eta^{\frac{1}{b+2}}\leq (\varepsilon+\delta)^{\frac{1}{30b^2}}.$$
	We fix $m\in \mathcal{M}.$  {The LDT can then be established assuming further restrictions of $\omega\in\Omega$ {(cf. \eqref{omgdefn}).} }
	
	As we will see later, the proof of LDT can be accomplished in three steps: 
	\begin{itemize}
	\item The first step deals with the small scales, for which we establish LDE for all $\omega\in\Omega$ and scales 
	$$N_0\leq N \leq c \log^{\frac{1}{\rho_1}} \frac{1}{\varepsilon+\delta},\ c=c(\rho_1)>0.$$
	for some $\rho_1>0.$  In this step we only employ the Neumann series argument.  
	\item  
	 For the  intermediate scales 
	$$c\log^{\frac{1}{\rho_1}} \frac{1}{\varepsilon+\delta}\leq   N\leq (\varepsilon+\delta)^{-\frac{1}{10^4db^4}},$$
	we are  able to establish LDE
	{for all $\omega\in\Omega$}. The proof of such  intermediate scales LDE  is based on  preparation type  theorem together with clustering  properties of the spectrum of the diagonal matrix $D$ (cf. \eqref {sublthm}, Theorem~\ref{clusthm}). 
	\item  For the large scales (i.e., $N>(\varepsilon+\delta)^{-\frac{1}{10^4db^4}}$) LDE, we will apply matrix-valued Cartan's estimates  and semi-algebraic sets theory  of \cite{BGS02}. {This step requires both the Diophantine condition and the {additional non-resonant  condition}  on $\omega$.}  More precisely, we denote 
	\begin{align}
\label{dcomg}{\rm DC}(N)&=\left\{\omega\in\Omega:\ |k\cdot\omega|\geq N^{-10^5db^5}\ {\rm for}\ \forall\ 0<|k|\leq 10N\right\}.
\end{align}
Then we  have for $\varepsilon=O(\delta)$, 
\begin{align}               
{\rm meas}\left(\Omega\setminus\left(\bigcap_{N>(\varepsilon+\delta)^{-\frac{1}{10^4db^4}}} {\rm DC}(N)\right)\right) \ll\delta^b, 
\end{align}
which motivates the usage of the exponent $10^5db^5$ in the Diophantine condition of $\omega.$  {The additional non-resonant condition (i.e., $\omega\in\widetilde \Omega_N$, cf. \eqref{tilomg})} involves a sub-exponential bound in the scale $N$,   which  is imposed so that the sublinear bound  on bad boxes along the $\Z^d\ni n$-direction holds true. 

	\end{itemize}
	



{We also need to introduce the definition of semi-algebraic sets. 
\begin{defn}
A set ${S}\subset \mathbb{R}^d$ is called {\it semi-algebraic} if it is a finite union of sets defined by a finite number of polynomial equalities and inequalities. More precisely, let $\{P_1,\cdots,P_s\}\subset\mathbb{R}[x_1,\cdots,x_d]$ be a family of real polynomials whose degrees are bounded by $\kappa$. A (closed) semi-algebraic set ${S}$ is given by an expression
\begin{equation}\label{smd1221}
	{S}=\bigcup\limits_{l}\bigcap\limits_{\ell\in\mathcal{L}_l}\left\{x\in\mathbb{R}^d: \ P_{\ell}(x)\varsigma_{l\ell}0\right\},
\end{equation}
where $\mathcal{L}_l\subset\{1,\cdots,s\}$ and $\varsigma_{l\ell}\in\{\geq,\leq,=\}$. Then we say that ${S}$ has degree at most $s\kappa$. In fact, the degree of ${S}$ which is denoted by $\deg {S}$, means the  smallest $s\kappa$ over all representations as in (\ref{smd1221}).
\end{defn}}
 The central theorem of this section is 
	\begin{thm}[LDT]\label{ldtthm}
	Fix  $\alpha\in{\rm DC}_\nu$ and $\theta_0\in {\rm DC}_\alpha((\varepsilon+\delta)^{-\frac{1}{10^4db^4}})$ $({\rm cf}. \ \eqref{alpdc}-\eqref{thedc})$. Then for $0<\varepsilon+\delta\ll1$ and ${\varepsilon=O(\delta)}$,  there exists  $\mathcal{M}\subset [2,3]$ with 
	$
	{\rm meas}([2,3]\setminus\mathcal{M})\leq (\varepsilon+\delta)^{\frac{1}{30b^2}}
	$	so that for $m\in\mathcal{M}$  the following holds.  Let $0<\rho_1<\rho_2<\rho_3<1$ satisfy $\rho_1< \rho_2/5$ and $\rho_2>\frac23$. Then for every $N\geq N_0(b,d, \rho_1,\rho_2,\rho_3, n^{(1)},\cdots, n^{(b)})\gg1$, there exists $\Omega_N$ satisfying 
	\begin{align*}
	\Omega_N&=\left\{
	\begin{array}{ll}
\Omega\ {({\rm cf.} \ \eqref{omgdefn})}& {\rm if} \ N_0\leq N\leq (\varepsilon+\delta)^{-\frac{1}{10^4db^4}},\\
{\rm DC}(N)\cap\widetilde\Omega_N\cap\Omega_{N_2}\ {({\rm cf}. \ \eqref{dcomg}\ {\rm and}\ \eqref{tilomg})}& {\rm if} N>(\varepsilon+\delta)^{-\frac{1}{10^4db^4}}, 
\end{array}\right.
	\end{align*}
	so that, if $\omega\in\Omega_N$, then ${H}(\sigma)$ satisfies the 
	$(\rho_1, \rho_2, \rho_3, \gamma_{N},N)$-LDE  with $\gamma_{N}=\gamma_{N_2}-N^{-\kappa}$ for some  $\kappa=\kappa(b,d,\rho_1,\rho_2,\rho_3)>0, N_2=N^c\ll N$, {where $\gamma>0$ is defined in \eqref{gamdefn}}, $\widetilde\Omega_N\subset\Omega$ is a semi-algebraic set of ${\rm meas}(\Omega\setminus\widetilde\Omega_N)\leq e^{-N^\zeta}$ and $\deg \widetilde\Omega_N\leq N^C$ for some $\zeta=\zeta(b,d,\rho_2)\in (0,1)$ and $C=C(b,d)>1.$
		\end{thm}

This  theorem  will be proved in detail in the following three sections.
	
\subsubsection{The small scales LDE}

This section is devoted to the proof of LDE for small scales, i.e., 
\begin{align*}
N_0\leq N\leq c\log^{\frac{1}{\rho_1}}\frac{1}{\varepsilon+\delta},
\end{align*}
 where $N_0\gg1$ depends on $b,d, \rho_1$ and $c\in(0,1)$. In this case we only use the standard Neumann series argument. As a result, we can establish LDE  for all $\omega\in\Omega.$ We have 

\begin{lem}\label{inilem}
Let $0<\rho_1<1$ and 
\begin{align*}
&\ \ \ \Sigma_N\\
&=\bigcup_{\Lambda\in (0,n)+\mathcal{ER}_0(N),\ |n|\leq 2N}\left\{\sigma\in\R:\ \min_{\xi=\pm1, (k,n')\in\Lambda}|\xi(\sigma+k\cdot\omega)+\mu_{n'}|\leq e^{-2N^{\rho_1}}\right\}.
\end{align*}
Then  for $N_0\leq N \leq 10^{-\frac{1}{\rho_1}} \log^{\frac{1}{\rho_1}}\frac{1}{\varepsilon+\delta}$, we have 
\begin{align}\label{inieq1}
{\rm meas}(\Sigma_N)\leq e^{-N^{\rho_1}},
\end{align}
where $N_0>0$ only depends on $b,d,\rho_1.$ Moreover, if $\sigma\notin \Sigma_N$, we have  for all $\Lambda\in (0,n)+\mathcal{ER}_0(N)$ with $|n|\leq 2N$,
\begin{align}
\label{inieq2}\|G_{\Lambda}(\sigma)\|&\leq 2e^{4N^{\rho_1}},\\
\label{inieq3}|G_{\Lambda}(\sigma)(j;j')|&\leq e^{4N^{\rho_1}}e^{-\gamma|j-j'|}\ {\rm for}\ j\neq j'.
\end{align}

\end{lem}
\begin{rem}
This lemma works for all  $\omega\in\Omega$ (cf. \eqref{omgdefn}). 
\end{rem}

{
\begin{proof}
We first establish the measure bound \eqref{inieq1}. Note that for each $\xi\in\{\pm1\}$ and $(k, n')\in\Lambda,$ the set 
$$\left\{\sigma\in\R:\ |\xi(\sigma+k\cdot\omega)+\mu_{n'}|\leq e^{-2N^{\rho_1}}\right\}$$
is an interval of length $2e^{-2N^{\rho_1}}$. 
So taking account of all $\Lambda\in (0,n)+\mathcal{ER}_0(N),  \xi\in\{\pm 1\}$ and $(k, n')\in\Lambda$,  we obtain for $N\geq N_0,$
\begin{align*}
{\rm meas}(\Sigma_N)&\leq C(b,d)N^{C(b,d)}e^{-2N^{\rho_1}}
\\
&\leq e^{-N^{\rho_1}},
\end{align*}
where in the last inequality we use $N\geq N_0(b,d,\rho_1)\gg1.$

Next, we let  $\sigma\notin \Sigma_N$ and fix $\Lambda\in(0,n)+\mathcal{ER}_0(N)$ with $|n|\leq 2N$. We will use the Neumann series  argument (cf. the proof of Theorem~4.3 in \cite{JLS20})   to control the Green's function $G_\Lambda(\sigma)$. More precisely,  denote  $P=\varepsilon  \Delta+\delta T_{\phi}.$ From   $|\Delta((k',n'); (k'',n''))|=\delta_{k',k''}\cdot \delta_{\|n'-n''\|_1,1}$ and the exponential decay of  $T_\phi$ (cf. \eqref{gamdefn}), it  follows that 
\begin{align}\label{pdecay}
|P((k',n'); (k'',n''))|\leq C(\varepsilon+\delta) e^{-\gamma|k'-k''|-\gamma|n'-n''|}. 
\end{align}
We define  $A=R_{\Lambda} D(\sigma)R_{\Lambda}$ (cf. \eqref{dsgmdefn}) and $B=R_{\Lambda}PR_\Lambda.$   Then $H_{\Lambda}(\sigma)=R_{\Lambda}H(\sigma)R_\Lambda=A+B$ since \eqref{hsgmdefn}. From $\sigma\notin\Sigma_N$ and $N\leq 10^{-\frac{1}{\rho_1}}\log^{\frac{1}{\rho_1}}\frac{1}{\varepsilon+\delta} $, we get 
\begin{align}\label{aestm}
\|A^{-1}\|\leq e^{4N^{\rho_1}}\leq {(\varepsilon+\delta)^{-\frac25}}.
\end{align}
Using \eqref{pdecay} and    the Schur's test, we obtain   for some $C=C(b,d)>0,$
\begin{align}\label{abest}
\|B\|\leq \sup_{(k',n')\in\Lambda}\sum_{(k'',n'')\in\Lambda} |P((k',n'); (k'',n''))|\leq C N^{b+d}(\varepsilon+\delta).
\end{align}
Combining \eqref{aestm} and \eqref{abest} implies  for $0<\varepsilon+\delta\ll1,$
$$\|A^{-1}B\|\leq C N^{b+d} (\delta+\varepsilon)^{\frac35}\leq (\varepsilon+\delta)^{\frac{2}{5}}<\frac12.$$
 Hence by applying the Neumann series argument, we have   
\begin{align*}
\|G_\Lambda(\sigma)\|&=\|(I+A^{-1}B)^{-1}A^{-1}\|=\|\sum_{l=0}^\infty(-A^{-1}B)^lA^{-1}\|\\
&\leq \sum_{l=0}^\infty 2^{-l}\|A^{-1}\|\leq 2\|A^{-1}\|\leq 2e^{4N^{\rho_1}},
\end{align*}
where $I$ denotes the identity operator on $\ell^2(\Lambda).$
So  we have proven \eqref{inieq2}.  Recalling again \eqref{pdecay} and \eqref{aestm}, we  obtain  for $l\geq 1,$
$$ |(A^{-1}B)^l(j;j')| \leq C^l N^{(b+d)l}(\varepsilon+\delta)^{\frac{3l}{5}} e^{-\gamma|j-j'|}.$$
Using again the  Neumann series argument  and since $A$ is diagonal, we  have  for $j\neq j'$ and $0<\varepsilon+\delta \ll1, $
\begin{align*}
 |G_{\Lambda}(\sigma)(j;j')|&\leq \sum_{l\geq 1} |(A^{-1}B)^l(j;j')|\cdot  |A^{-1}(j';j')|\\
 &\leq e^{-\gamma|j-j'|}  \sum_{l\geq 1} C^l N^{(b+d)l}(\varepsilon+\delta)^{\frac{3l}{5}}\|A^{-1}\|\\
 &\leq e^{-\gamma|j-j'|} \sum_{l\geq 1} (\varepsilon+\delta)^{\frac{2l}{5}}  \|A^{-1}\| \leq  \|A^{-1}\| \cdot e^{-\gamma|j-j'|}. 
\end{align*}
 So  \eqref{inieq3} follows.   

\end{proof}

}

\subsubsection{The intermediate scales LDE}
This section and the next deal with intermediate scales LDE.  We will show that the LDE hold true for all $\omega\in\Omega$  and scales 
$$c\log^{\frac{1}{\rho_1}}\frac{1}{\varepsilon+\delta}\leq N\leq (\varepsilon+\delta)^{-\frac{1}{10^4db^4}}, \ c>0.$$  
In this case the analysis in section~\ref{NRomega0} becomes essential.  
\begin{lem}\label{intlem}
Let  $\alpha\in{\rm DC}_\nu\subset {\rm DC}(L;L^{-3d})$ and let $\theta_0$,  $\mathcal{M}\subset [2, 3]$ be given by Theorem \ref{clusthm} with 
$$ \eta=(\varepsilon+\delta)^{\frac{1}{8b}}, \ L=(\varepsilon+\delta)^{-\frac{1}{10^4db^4}}.$$
Assume $m\in \mathcal{M}$.   Fix  $1>\rho_3>\rho_2>2\rho_1>0, \rho_2>\frac12$ and $$c \log^{\frac{1}{\rho_1}}\frac{1}{\varepsilon+\delta}\leq N\leq (\varepsilon+\delta)^{-\frac{1}{10^4db^4}}.$$ Then for each $\omega\in\Omega\ ({\rm cf}. \ \eqref{omgdefn})$, 
there exists a set $\Sigma_N\subset\R$ with 
\begin{align*}
{\rm meas}(\Sigma_N)\leq (\varepsilon+\delta)^{-\frac{9}{10}}e^{-\frac{1}{2b}N^{\frac{\rho_2}{2}}},
\end{align*}
so that if $\sigma \notin\Sigma_N$, then   for all $\Lambda\in (0,n')+\mathcal{ER}_0(N)$ with $|n'|\leq 2N$, one has 
\begin{align*}
\|G_{\Lambda}(\sigma)\|&\leq e^{N^{\rho_2}},\\
|G_{\Lambda}(\sigma)(j;j')|&\leq e^{-\gamma_N|j-j'|}\ {\rm for}\ |j-j'|\geq N^{\rho_3},
\end{align*}
where 
$$\gamma_N=\gamma-\log^{-2}\frac{1}{\varepsilon+\delta}-N^{-\kappa}$$
with $\kappa=\kappa(b,d,\rho_2,\rho_3)>0.$
\end{lem}
\begin{rem}
Note that this lemma also works for  all  $\omega\in\Omega$. Its proof is somewhat independent of the other parts of the paper and can be read separately.
\end{rem}

\subsubsection{Proof of the intermediate scales LDE}
The proof can be decomposed into two steps.  {The first step aims to construct $\Sigma_N$ and control ${\rm meas}(\Sigma_N)$  so that  $\|G_\Lambda(\sigma)\|\leq e^{N^{\rho_2}}$ for $\sigma\notin\Sigma_N.$   For this purpose, we first use the Neumann series argument (similar to the proof of Lemma \ref{inilem}) to cover $\Sigma_N$ with intervals of small length. Next, on each of these small  intervals, we combine  the counting bound related to   $\omega^{(0)}$ (cf. \eqref{sublthm}), the Schur's complement argument and the Rouch\'e's theorem  to remove singular   $\sigma$ (i.e.,  $\sigma$ so that $\|G_\Lambda(\sigma)\| >e^{({\rm diam} \Lambda )^{\rho_2}}$).  This  second reduction is necessary since $N\geq c\log^{\frac{1}{\rho_1}}\frac{1}{\varepsilon+\delta}$,  and it  shows  that  the singular $\sigma$ must  belong to  a finitely  (at most $2b$) many subintervals with much smaller  length.    To apply the Schur's complement reduction,  we again use the Neumann series argument.}  The second step aims to establish exponential off-diagonal decay of Green's functions, which relies on  an  application of the coupling lemma (in the spirit of iterating the resolvent identities)  in  \cite{HSSY}. \\
{\bf Step 1.}  
We  define 
\begin{align}\label{intsgmn}
\Sigma_N=\bigcup_{\sqrt{N}\leq N'\leq N,\ \Lambda\in (0,n')+\mathcal{ER}_0(N'),\ |n'|\leq 2N'}\Sigma_\Lambda,
\end{align}
where
\begin{align*}
\Sigma_\Lambda=\left\{\sigma\in\R:\  \|G_\Lambda(\sigma)\|\geq  e^{({\rm diam}\ \Lambda)^{\rho_2}}\right\}.
\end{align*}
We first show that 
\begin{align}\label{set}
\Sigma_\Lambda\subset \bigcup_{\xi=\pm 1, i=(k,n)\in\Lambda} I_{\xi, i},
\end{align}
where 
\begin{align*}
I_{\xi, i}=\left\{\sigma\in\R:\ |\xi(\sigma+k\cdot\omega^{(0)})+\mu_n|\leq \frac{1}{100}(\varepsilon+\delta)^{\frac{1}{7b}}\right\}.
\end{align*}
In fact, if $\sigma\notin  \bigcup\limits_{\xi=\pm 1, i=(k,n)\in\Lambda} I_{\xi, i},$ then for each $\omega\in\Omega\ ({\rm cf.}\ \eqref{omgdefn}),$
\begin{align*}
&\ \ \ \ \inf_{\xi=\pm 1, i=(k,n)\in\Lambda} |\xi(\sigma+k\cdot\omega)+\mu_n|\\
&\geq \inf_{\xi=\pm 1, i=(k,n)\in\Lambda} |\xi(\sigma+k\cdot\omega^{(0)})+\mu_n|-\sup_{k\in\Lambda}|k\cdot(\omega-\omega^{(0)})|\\
&\geq \frac{1}{100}(\varepsilon+\delta)^{\frac{1}{7b}}-C(\varepsilon+\delta)^{-\frac{1}{10^4db^4}}\delta\\
&\geq \frac{1}{200}(\varepsilon+\delta)^{\frac{1}{7b}}.
\end{align*}
Since $\rho_2>2\rho_1$ and $N\geq c \log^{\frac{1}{\rho_1}} \frac{1}{\varepsilon+\delta}$, using the Neumann series argument  (similar to the proof of Lemma \ref{inilem})  implies  that
\begin{align*}
\|G_\Lambda(\sigma)\|\leq C(\varepsilon+\delta)^{-\frac{2}{7b}}<e^{({\rm diam}\ \Lambda)^{\rho_2}}
\end{align*}
when  $\sigma\notin  \bigcup\limits_{\xi=\pm 1, i=(k,n)\in\Lambda} I_{\xi, i}.$ This proves \eqref{set} and 
\begin{align*}
\Sigma_\Lambda=\bigcup_{\xi=\pm 1, i\in\Lambda}(\Sigma_\Lambda\cap I_{\xi,i}).
\end{align*}

In the following we estimate ${\rm meas}(\Sigma_\Lambda\cap I_{\xi,i}).$ For this purpose, we  fix any $\xi^*\in\{\pm1\}$, $i^*=(k^*,n^*)\in\Lambda$ and denote 
$$ I^*=I_{\xi^*, i^*}=\left[\sigma^*-\frac{1}{100}(\varepsilon+\delta)^{\frac{1}{7b}}, \sigma^*+\frac{1}{100}(\varepsilon+\delta)^{\frac{1}{7b}}\right],$$
where 
\begin{align}\label{sgm*}
\sigma^*=-k^*\cdot\omega^{(0)}-\frac{1}{\xi^*}\mu_{n^{*}}.
\end{align}
We will study $G_\Lambda(\sigma)$ for $\sigma\in I^*.$ From (4) of Theorem \ref{clusthm} and \eqref{etaed}, we have for  $\sigma^*\in\R$,  there exists $B_*=B_*^+\cup B_*^-\subset \Lambda$ with $\# B_*^{\pm}\leq b$ so that
\begin{align*}
\inf_{\xi=\pm1, i=(k,n)\in\Lambda\setminus B_*}|\xi(\sigma^*+k\cdot\omega^{(0)})+\mu_n|\geq \frac{(\varepsilon+\delta)^{\frac{1}{8b}}}{2}.
\end{align*}
Define for $r>0$ and $z_0\in \C$ the set ${\mathbb D}_r(z_0)=\{z\in\C:\ |z-z_0|\leq r\}$.  We claim  that if 
$z\in \mathbb{D}_{(\varepsilon+\delta)^{\frac{1}{8b}}/4}(\sigma^*),$ 
then
\begin{align}\label{1221lbd}
\inf_{\xi=\pm1, i=(k,n)\in\Lambda\setminus B_*}|\xi(z+k\cdot\omega)+\mu_n|\geq \frac{(\varepsilon+\delta)^{\frac{1}{8b}}}{5}.
\end{align}
In fact,  we have
\begin{align*}
&\ \ \ \ \inf_{\xi=\pm 1, i=(k,n)\in\Lambda\setminus B_*} |\xi(z+k\cdot\omega)+\mu_n|\\
&\geq \inf_{\xi=\pm 1, i=(k,n)\in\Lambda\setminus B_*} |\xi(\sigma^*+k\cdot\omega^{(0)})+\mu_n|-\sup_{k\in\Lambda\setminus B_*}|k\cdot(\omega-\omega^{(0)})|-|z-\sigma^*|\\
&\geq \frac{1}{2}(\varepsilon+\delta)^{\frac{1}{8b}}-C(\varepsilon+\delta)^{-\frac{1}{10^4db^4}}\delta-\frac{1}{4}(\varepsilon+\delta)^{\frac{1}{8b}}\\
&\geq \frac{1}{5}(\varepsilon+\delta)^{\frac{1}{8b}}.
\end{align*}
Recalling \eqref{hsgmdefn} and \eqref{dsgmdefn}, we let  $H(z)$ and  $D(z)$  be  the  analytic extensions  to   $z\in\C$.  {Denote further $H_\Lambda(z)=R_\Lambda H(z)R_\Lambda$ and  $D_\Lambda(z)=R_\Lambda D(z)R_\Lambda$}. So we have for $z\in\mathbb{D}_{(\varepsilon+\delta)^{\frac{1}{8b}}/4}(\sigma^*),$   
\begin{align*}
\sup_{j,j'\in\Lambda,\ z\in\mathbb{D}_{(\varepsilon+\delta)^{\frac{1}{8b}}/4}(\sigma^*)}|{D}_\Lambda(z)(j;j')|\leq C(b)(\varepsilon+\delta)^{-\frac{1}{10^4db^4}}. 
\end{align*}
 More importantly, if we denote $\Lambda^c=\Lambda\setminus B_*$ and $G_{\Lambda^c}(z)=(R_{\Lambda^c}H(z)R_{\Lambda^c})^{-1}$, then by  \eqref{1221lbd} and the  Neumann series argument,  we have  for all $z\in \mathbb{D}_{(\varepsilon+\delta)^{\frac{1}{8b}}/4}(\sigma^*),$
\begin{align*}
\|G_{\Lambda^c}(z)\|\leq 50(\varepsilon+\delta)^{-\frac{1}{4b}}.
\end{align*}
In the following we estimate $G_{\Lambda}(z)$ using the Schur's  complement reduction.  We define 
\begin{align*}
S(z)=R_{B_*}{H}(z)R_{B_*}-R_{B_*}{H}(z)R_{\Lambda^c}G_{\Lambda^c}(z)R_{\Lambda^c}{H}(z)R_{B_*}.
\end{align*}
We have
\begin{align}
\nonumber R_{B_*}{H}(z)R_{B_*}&=R_{B_*}{D}(z)R_{B_*}+R_{B_*}(\varepsilon\Delta+\delta{T}_\phi )R_{B_*}\\
\label{schureq1}&=R_{B_*}{D}(z)R_{B_*}+O(\varepsilon+\delta)
\end{align}
and 
\begin{align}\label{schureq2}
R_{B_*}{H}(z)R_{\Lambda^c}G_{\Lambda^c}(z)R_{\Lambda^c}{H}(z)R_{B_*}=O((\varepsilon+\delta)^{2-\frac{1}{4b}}).
\end{align}
Now we consider  the determinant  $\det S(z)$. Combining \eqref{schureq1},  \eqref{schureq2}  and $\omega=\omega^{(0)}+O(\delta)$  yields 
\begin{align*}
&\ \ \ \det S(z)\\
&= \prod_{ i=(k,n)\in B_*^+}(z+k\cdot\omega^{(0)}+\mu_n)\prod_{i=(k,n)\in B_*^-}(-z-k\cdot\omega^{(0)}+\mu_n)\\
&\ \ \ +O((\varepsilon+\delta)^{3/4})\\
&{=(-1)^{\#B_*^-}}\prod_{l=1}^{\# B_*}(z-\sigma_l)+O((\varepsilon+\delta)^{3/4}),
\end{align*}
where  for each $1\leq l\leq \#B_*,$
{$$\sigma_l\in \{-k\cdot\omega^{(0)}-\mu_n\}_{(k,n)\in B_*^+}\cup\{-k\cdot\omega^{(0)}+\mu_n\}_{(k,n)\in B_*^-}.$$}
From  the definition of $B_*$, we have  for all $z\in \mathbb{D}_{(\varepsilon+\delta)^{\frac{1}{8b}}/4}(\sigma^*), $
\begin{align*}
|z-\sigma_l|\leq \frac45(\varepsilon+\delta)^{\frac{1}{8b}}.
\end{align*}
From \eqref{sgm*}, there exists at least one $\sigma_l$ so that $|\sigma_l-\sigma^*|=O((\varepsilon+\delta)^{\frac45}).$

At this stage,  we need a useful result. 
\begin{lem}\label{pigelem}
{ Let $\varphi=\left(\frac{14}{15}\right)^{\frac{1}{5b}}<1$. 
There is $0\leq l_*\leq 5b-1$ so that
\begin{align*}
\left(\mathbb{A}_{l_*+1}\setminus\mathbb{A}_{l_*}\right)\cap \{\sigma_l\}_{1\leq l\leq \# B_*}=\emptyset, 
\end{align*}
where for $0\leq l\leq 5b$,
\begin{align*}
\mathbb{A}_l=\mathbb{D}_{(\varepsilon+\delta)^{\frac{\varphi^{l}}{7b}}}(\sigma^*).
\end{align*}
}
\end{lem}
\begin{proof}[Proof of Lemma \ref{pigelem}]
It suffices to note that 
\begin{align*}
\mathbb{A}_{5b}\setminus\mathbb{A}_{0}=\bigcup_{l=0}^{5b-1}\left(\mathbb{A}_{l+1}\setminus \mathbb{A}_{l}\right). 
\end{align*}
The proof then follows from the pigeonhole principle since $\#B_*\leq 2b$. 
\end{proof}
Let $r_*=(\varepsilon+\delta)^{\frac{\varphi^{l_*}}{7b}}$, where $\varphi$ and $l_*$ are defined in Lemma \ref{pigelem}. From Lemma \ref{pigelem}, we can restrict our considerations on  $\mathbb{A}_*$ with 
\begin{align*}
\mathbb{A}_*=\mathbb{D}_{\frac{r_*+r_*^\varphi}{2}}(\sigma^*). 
\end{align*}
Then
\begin{align*}
\mathbb{D}_{(\varepsilon+\delta)^{\frac{1}{7b}}}(\sigma^*)\subset\mathbb{A}_{l_*}\subset\mathbb{A}_*\subset \mathbb{D}_{(\varepsilon+\delta)^{\frac{1}{7.5b}}}(\sigma^*).
\end{align*}
Let 
\begin{align*}
 K=\{1\leq l\leq \# B_*:\  \sigma_l\in \mathbb{D}_{10r_*}(\sigma^*)\}.
\end{align*}

We have $K\neq \emptyset$ and for $l\not\in K$,
\begin{align*}
\sigma_l\notin\mathbb{A}_*,\ {\rm dist}(\sigma_l, \partial\mathbb{A}_*)\geq (\varepsilon+\delta)^{\frac{1}{7b}},
\end{align*}
where $\partial (\cdot)$ denotes the boundary of a set. 
As a result, we can write for $z\in\mathbb{A}_*,$
\begin{align}\label{Rouceq1}
\frac {\det S(z)}{\prod_{l\notin K}(z-\sigma_l)}= (-1)^{\#B_*^-}{\prod_{l\in K}(z-\sigma_l)}+O((\varepsilon+\delta)^{\frac13}).
\end{align}
We have 
\begin{align*}
\min_{z\in \partial\mathbb{A}_*}{\prod_{l\in K}|z-\sigma_l|} \geq (\frac1{100b})^{\#K}(\varepsilon+\delta)^{\frac{\#K}{8b}}\geq c(\varepsilon+\delta)^{1/4}\gg C(\varepsilon+\delta)^{\frac13}.
\end{align*}
Using the Rouch\'e's theorem shows that  the function $\frac {\det S(z)}{\prod_{l\notin K}(z-\sigma_l)}$ has exactly $K$ zeros (denoted by $z_l,\ l\in K$)  in $\mathbb{A}_*.$   
So on $\mathbb{A}_*$, we have 
\begin{align}\label{Rouceq3}
\frac {\det S(z)}{\prod_{l\notin K}(z-\sigma_l)}=g(z)\prod_{l\in K}(z-z_l),
\end{align}
where $g$ is analytic on $\mathbb{A}_*$ with $\inf\limits_{z\in\mathbb{A}_*}|g(z)|>0.$ Using \eqref{Rouceq3}, we have  for all $z\in\partial \mathbb{A}_*$ and some $C=C(b)>1,$
\begin{align*}
C^{-1}\leq |g(z)| \leq C
\end{align*}
assuming $0<\varepsilon+\delta\leq c(b)\ll1.$
Then using the maximum principle, we get for all $z\in  \mathbb{A}_*,$
\begin{align*}
C^{-1}\leq |g(z)| \leq C. 
\end{align*}
In particular, one has for all $z\in\mathbb{A}_*,$
\begin{align*}
|\det S(z)|&\geq C(b) (\varepsilon+\delta)^{\frac{\#B_*-\#K}{7b}}\prod_{l\in K}|z-z_l|\\
&\geq  (\varepsilon+\delta)^{\frac{2}{7}}\prod_{l\in K}|z-z_l|
\end{align*}
assuming $0<\varepsilon+\delta\leq c(b)\ll1.$ Thus using  Hadamard's inequality and Cramer's rule, we get    for $z\in\mathbb{A}_*,$
\begin{align*}
\|S^{-1}(z)\|&\leq C(b)    (\varepsilon+\delta)^{-\frac{2b}{10^4db^4}}(\varepsilon+\delta)^{-\frac{2}{7}}\prod_{l\in K}|z-z_l|^{-1}\\
&\leq (\varepsilon+\delta)^{-\frac{1}{3}}\prod_{l\in K}|z-z_l|^{-1}.
\end{align*}
Applying the Schur's complement argument  (cf. Lemma B.1 in \cite{CSZ24}) shows that for $z\in \mathbb{A}_*,$
\begin{align*}
\|G_\Lambda(z)\|&\leq 4(1+\|G_{\Lambda^c}(z)\|)^2(1+\|S^{-1}(z)\|)\\
&\leq (\varepsilon+\delta)^{-\frac{7}{8}}\prod_{l\in K}|z-z_l|^{-1}.
\end{align*}
Recall that $I^*\subset \mathbb{A}_*$. So for $\sigma\in I^*,$ we have
\begin{align*}
\|G_\Lambda(\sigma)\|&\leq (\varepsilon+\delta)^{-\frac{7}{8}}\prod_{l\in K}|\sigma-z_l|^{-1} 
\end{align*}
and 
\begin{align*}
\Sigma_\Lambda\cap I^*\subset \bigcup_{l\in K}\left\{\sigma\in I^*:\ |\sigma-z_l|\leq (\varepsilon+\delta)^{-\frac{7}{8}}e^{-\frac{1}{2b}({\rm diam}\ \Lambda)^{\rho_2}}\right\}.
\end{align*}
This implies 
\begin{align*}
{\rm meas}(\Sigma_\Lambda\cap I^*)\leq C(\varepsilon+\delta)^{-\frac{7}{8}}e^{-\frac{1}{2b}({\rm diam}\ \Lambda)^{\rho_2}}
\end{align*}
and thus 
\begin{align*}
{\rm meas}(\Sigma_\Lambda)\leq (\varepsilon+\delta)^{-\frac{8}{9}}e^{-\frac{1}{2b}({\rm diam}\ \Lambda)^{\rho_2}},
\end{align*}
which combined with \eqref{intsgmn}  shows 
\begin{align*}
{\rm meas}(\Sigma_N)\leq (\varepsilon+\delta)^{-\frac{9}{10}}e^{-\frac{1}{2b}N^{\frac{\rho_2}{2}}}.
\end{align*}

\noindent {\bf Step 2.} In this step we prove the exponential off-diagonal decay of $G_\Lambda(\sigma)$ assuming $\sigma\not\in\Sigma_{N}$ and $\Lambda\in(0,n')+\mathcal{ER}_0(N)$ with $|n'|\leq 2N.$  By Theorem \ref{clusthm}, we know that there exists some  $B\subset\Lambda_L\cap \Lambda$ so that  $\#B\leq 2b$ and 
\begin{align*}
\inf_{\xi=\pm1, (k,n)\in\Lambda\setminus B}|\xi(\sigma+k\cdot\omega)+\mu_n|\geq \frac{(\varepsilon+\delta)^{\frac{1}{8b}}}{5}.
\end{align*}
Using again the Neumann series argument  implies that for any $\Lambda'\subset\Lambda$ with $\Lambda'\cap B=\emptyset,$ one has 
\begin{align}
\label{intexp1}\|G_{\Lambda'}(\sigma)\|&\leq 50(\varepsilon+\delta)^{-\frac{1}{4b}},\\
\label{intexp2}|G_{\Lambda'}(\sigma)(j;j')|&\leq  (\varepsilon+\delta)^{-\frac{1}{4b}}e^{-\gamma|j-j'|}\ {\rm for}\ j\neq j'.
\end{align}
We say that $\mathcal{ER}(\sqrt{N})\ni Q\subset\Lambda$ is $\sqrt{N}$-regular if both \eqref{intexp1} and \eqref{intexp2} hold with $\Lambda'=Q.$  Otherwise, $Q$ is called $\sqrt{N}$-singular.  Let $\mathcal{F}$ be any family of pairwise disjoint $\sqrt{N}$-singular sets   contained in $\Lambda$. We obtain 
\begin{align}\label{1221subl}
\#\mathcal{F}\leq 2b. 
\end{align}
 Since $\sigma\not\in \Sigma_N$, we have for all $\Lambda'\in \mathcal{ER}(L), \Lambda'\subset \Lambda$ with $\sqrt{N}<L<N$, 
\begin{align}\label{1221glb}
\|G_{\Lambda'}(\sigma)\|\leq e^{L^{\rho_2}}. 
\end{align}
{We are ready  to use the coupling lemma (i.e., Lemma 4.2) of \cite{HSSY} to get the exponential off-diagonal decay of $G_\Lambda(\sigma)$.  So it needs to verify all assumptions of Lemma 4.2 \cite{HSSY}:  The assumption  (i) of Lemma 4.2  follows from \eqref{pdecay}; the assumption (ii) of Lemma 4.2  follows from  \eqref{intexp1}, \eqref{intexp2} and \eqref{1221glb}; the assumption (iii) of Lemma 4.2  follows from \eqref{1221subl}.  
Finally,  applying  Lemma 4.2 of \cite{HSSY}  with $T=H_{\Lambda}(\sigma), \tau=\frac12, b=\rho_2>\frac12, \theta=\rho_3, \rho=\alpha=\gamma$ shows 
\begin{align*}
|G_\Lambda(\sigma)(j;j')|\leq e^{-\gamma_N|j-j'| } \ {\rm for}\  |j-j'|>N^{\rho_3},
\end{align*}
where 
\begin{align*}
\gamma_N=\gamma-\log^{-2}\frac{1}{\varepsilon+\delta}-N^{-\kappa},\ \kappa=\kappa(b, d, \rho_2,\rho_3)>0.
\end{align*}
This proves the exponential off-diagonal decay. }

We  have completed the proof of Lemma \ref{intlem}. 
\hfill $\square$


	\subsubsection{The large scales LDE}
	In this section we will prove the LDE for 
	\begin{align*}
	(\varepsilon+\delta)^{-\frac{1}{10^4db^4}}<N<\infty.
	\end{align*}
We  remark that since the operator ${D}+\varepsilon \Delta$ is \textit{non-diagonal}  in the  $n\in\Z^d$ direction, 
to establish LDT,  it requires fine properties of such  operators projected on $\Z^d$. 
{So denote by $\mathcal{ER}_d(M)$ (resp. $\mathcal{ER}_{d,0}(M)$) the set of all elementary $M$-regions (resp. $M$-regions of the center being $0$) on $\Z^d$}. 
For $E\in\R$ and $Q\subset\Z^d,$ we let 
\begin{align}\label{TQ}
\widetilde T_Q(E;\theta)= R_Q(\cos(\theta+n\cdot\alpha)\delta_{n,n'}+m-E+\varepsilon\Delta)R_Q.
\end{align}

We have
\begin{lem}\label{Boulem}
Fix $\alpha\in{\rm DC}_{\nu}$. Let  $0<\varepsilon\ll1$ and $N\geq N_0\gg1.$ 
Then 
for all $E\in \R$, there is a set $\Theta_N=\Theta_N(E)\subset\R$ with ${\rm meas}(\Theta_N)\leq e^{-N^{\rho_4}}$, so that if $\theta\notin\Theta_N,$ then for all $Q\in\mathcal{ER}_{d,0}(N),$
\begin{align*}
\|\widetilde T^{-1}_{Q}(E;\theta)\|&\leq e^{\sqrt{N}},\\
|\widetilde T^{-1}_{Q}(E;\theta)(n;n')|&\leq e^{-\frac12|\log\varepsilon|\cdot|n-n'|}\ {\rm for}\ |n-n'|\geq {N}^{\rho_3},
\end{align*}
where $0<\rho_4\ll1.$

\end{lem}
\begin{rem}
For fixed $\alpha\in{\rm DC}_\nu$, the LDT of this type can be extended to more general  quasi-periodic Schr\"odinger operators with non-constant analytic potentials defined on $\T$. 
\end{rem}



To perform the {multi-scale analysis (MSA)}   scheme, we need two small scales $N_1<N_2<N$ with
\begin{align*}
N_2=N_1^{\widetilde C},\  N=N_2^{\widetilde C},
\end{align*}
where $\widetilde C>1$ will be specified below.  

In the following, we always assume $\omega\in\Omega_{N_2}\subset\Omega_{N_1}.$

The large scales LDE are  as follows. 
\begin{lem}\label{lsclem}
Let $\alpha, \theta_0,   m$ be given by  Lemma \ref{intlem} and let $N$ satisfy  
\begin{align*}
	(\varepsilon+\delta)^{-\frac{1}{10^4db^4}}<N<\infty.
	\end{align*}
Then there is a semi-algebraic set $\widetilde\Omega_N\subset\Omega\ ({\rm cf.}\ \eqref{omgdefn})$ with ${\rm meas}(\Omega\setminus\widetilde \Omega_N)\leq e^{-N^\zeta}$  and $\deg \widetilde \Omega_N\leq N^C$ for some $\zeta=\zeta(b,d, \rho_2)\in(0,1), C=C(b,d)>1$, so that the following holds for $\omega\in{\rm DC}(N)\cap\widetilde\Omega_N\cap\Omega_{N_2}\ ({\rm cf.}\ \eqref{dcomg}).$ Let $0<\rho_1<\rho_2<\rho_3<1$ satisfy $\rho_2>\frac 23$ and $\rho_1<\rho_2/5$.   There is a set
$\Sigma_N\subset\R$ with ${\rm meas}(\Sigma_N)\leq e^{-N^{\rho_1}}$, so that if $\sigma\notin\Sigma_N$, then for all $\Lambda\in(0,n)+\mathcal{ER}_0(N)$ with $|n|\leq 2N,$
\begin{align*}
\|G_{\Lambda}(\sigma)\|&\leq e^{N^{\rho_2}},\\
|G_{\Lambda}(\sigma)(j;j')|&\leq e^{-\gamma_N|j-j'|}\ {\rm for}\ |j-j'|\geq N^{\rho_3},
\end{align*}
where $\gamma_N=\gamma_{N_2}-N^{-\kappa}$ for some $\kappa=\kappa(b,d,\rho_1,\rho_2,\rho_3)>0$ and $N_2=N^{ c}$  $(c\in (0,1))$. 
\end{lem}

For $\Lambda\in \mathcal{ER}(K)$, we say $\Lambda$ is $\sigma$-good if 
\begin{align}
\label{1221K1}\|G_{\Lambda}(\sigma)\|&\leq e^{K^{\rho_2}},\\
\label{1221K2}|G_{\Lambda}(\sigma)(j;j')|&\leq e^{-\gamma_K |j-j'|}\ {\rm for}\ |j-j'|\geq K^{\rho_3},
\end{align}
where $\gamma_K>0.$

 We need the following ingredients to prove Lemma \ref{lsclem}. 



\begin{lem}\label{lsclem1}
 \textit{ There is $ C>1$ depending  only on  $b,d$  so that  if $\widetilde C^2 \geq \frac {C}{\rho_2}$,  then the number of  $\sigma$-bad $\Lambda\in (k,n)+\mathcal{ER}_0(N_1)$ satisfying $|k|\leq N$ and $|n|\leq 2N_1$ is at most ${N^{\rho_2/4}}/2.$ }
 \end{lem}
\begin{proof}[Proof of Lemma~\ref{lsclem1}]
First, we define $\widetilde \Sigma_{N_1}$ to be the set of  $\sigma\in\R$ so that,  there is  $\Lambda\in \mathcal{ER}_0(N_1)$  such that $\Lambda$ is $\sigma$-bad.  We then obtain 
\begin{align*}
{\rm meas}(\widetilde \Sigma_{N_1})\leq CN_1^Ce^{-N_1^{\rho_1}}<e^{-\frac12 N_1^{\rho_1}}.
\end{align*}
{Following \cite{Bou07} (cf. page 691),  we replace  \eqref{1221K1}--\eqref{1221K2} (with  $K=N_1$)   by essentially equivalent statements
\begin{align}
\label{1221semi1}\sum_{j,j'\in\Lambda}|G_\Lambda(\sigma)(j;j')|^2&< e^{2N_1^{\rho_2}},\\
\label{1221semi2}\sum_{|j-j'|\geq N_1^{\rho_3}}|G_\Lambda(\sigma)(j;j')|^2&< e^{-2\gamma_{N_1}N_1}. 
\end{align}
Then we express the matrix elements $G_\Lambda(\sigma)(j;j')$ by Cramer’s rule, so that
both \eqref{1221semi1} and \eqref{1221semi2} may be written in the form $P'(\sigma)>0$,  where $P'$  is a
polynomial of degree at most $N_1^{C(b,d)}$ since  \eqref{dsgmdefn}. This allows us to  regard  $\widetilde \Sigma_{N_1}$ as a semi-algebraic set of degree $\deg \widetilde \Sigma_{N_1}\leq N^{C}$ with $C=C(b,d)>0.$}  Then using Basu-Pollack-Roy Theorem \cite{BPR96} on Betti numbers of a semi-algebraic set (cf. also Proposition 9.2 of \cite{Bou05}), we have a decomposition of $$\widetilde \Sigma_{N_1}=\bigcup_{1\leq l\leq N_1^{C}}I_l,$$
where each $I_l$ is an interval of length $|I_l|\leq e^{-\frac12 N_1^{\rho_1}}.$  Now assume that $\Lambda\in(k,n)+\mathcal{ER}_0(N_1)$ and $\Lambda'\in(k',n)+\mathcal{ER}_0(N_1)$ are all $\sigma$-bad for  $|n|\leq 2N_1, k\neq k'$ and $|k|, |k'|\leq N$.  Then from Toeplitz property of ${H}(\sigma)$ in the $k$-direction, we obtain 
$$\sigma+k\cdot\omega,\ \sigma+k'\cdot\omega\in\widetilde \Sigma_{N_1}. $$
As a result, by $k\neq k'$ with $|k|, |k'|\leq N$ and $\omega\in{\rm DC}(N)\ ({\rm cf.}\ \eqref{dcomg})$,  we have
\begin{align*}
{e^{-\frac12 N_1^{\rho_1}}}\ll N^{-10^5db^5}\leq |(k-k')\cdot\omega|,
\end{align*}
which implies $k$ and $k'$ can not stay in the same interval $I_l$ for   $1\leq l\leq N_1^{C}.$  Thus we have shown that the number of $\sigma$-bad $\Lambda\in (k,n)+\mathcal{ER}_0(N_1)$ with $|n|\leq 2N_1$ and $|k|\leq N$ is at most
\begin{align*}
CN_1^{C}N_1^d\leq N_1^{C}=N^{\frac{C}{\widetilde C^2}}\leq N^{\rho_2/4}/2 
\end{align*}
assuming $\widetilde C^2\geq \frac{C}{\rho_2}. $ This proves Lemma \ref{lsclem1}. 
\end{proof}

 
Next, for $\Lambda\subset\Z^{b+d}$, write  $\Lambda=\Pi_b\Lambda\times \Pi_d\Lambda$ with $\Pi_b\Lambda$ (resp. $\Pi_d\Lambda$) denoting  the projection of $\Lambda$ on $\Z^b$ (resp. $\Z^d$). For each $k\in \Pi_b\Lambda, $ define 
$$\Pi_{d}\Lambda(k)=\{n:\ (k,n)\in\Lambda\}\subset\Z^d.$$
We have 
 \begin{lem}\label{lsclem2}
 { 
\begin{itemize}
\item[(1)] Fix $\Lambda\in (k',n')+\mathcal{ER}_0(N_1) $ with  $|n'|>2N_1$.  Then there is a collection of $\{\zeta_l\geq \frac12\}_{1\leq l\leq N_1^C}$ depending only on $\Lambda, \alpha, \theta_0, m$ so that if 

$$\min_{k\in\Pi_b\Lambda, 1\leq l\leq N_1^C, \xi=\pm1}|\sigma+k\cdot\omega+\xi\sqrt{\zeta_l}|> e^{-\frac14 N_1^{\rho_2/2}},$$
 then $\Lambda$ is $\sigma$-good. 
\item[(2)]There exists some $\Sigma_{N_2}\subset \R$ with ${\rm meas}(\Sigma_{N_2})\leq  e^{-\frac15 N_2^{\rho_2/2}}$ so that if $\sigma\notin\Sigma_{N_2}$,  then each $\Lambda\in (k,n)+\mathcal{ER}_0(N_2) $ with $|(k,n)|\leq 10N$ and $|n|>2N_2$ is $\sigma$-good. 
\end{itemize}}
  \end{lem}
\begin{proof}[Proof of Lemma \ref{lsclem2}]
\begin{itemize}
\item[(1)] 
Fix  $k\in\Pi_b\Lambda$.  Let $\Lambda'\subset \Pi_d\Lambda(k)$
satisfy $\Lambda'\in \mathcal{ER}_d(L)$ for some $\sqrt{N_1}\leq L\leq N_1.$
Write 
\begin{align*}
 A_k=R_{\Lambda'}({D}(\sigma)+\varepsilon\Delta)R_{\Lambda'}.
\end{align*}
Then  the set of all eigenvalues of $A_k$ is given by 
$$\{-(\sigma+k\cdot\omega)^2+\zeta_l(k)\}.$$
So we have 
\begin{align*}
\|A_k^{-1}\|&\leq \max_{1\leq l\leq \#\Lambda'}|-(\sigma+k\cdot\omega)^2+\zeta_l(k)|^{-1}\\
&= \max_{1\leq l\leq \#\Lambda'}|\sigma+k\cdot\omega-\sqrt{\zeta_l(k)}|^{-1}\cdot|\sigma+k\cdot\omega+\sqrt{\zeta_l(k)}|^{-1}.
\end{align*}
Now assume that 
\begin{align*}
\min_{\sqrt{N_1}\leq L\leq N_1,  \mathcal{ER}_d(L)\ni\Lambda'\subset \Pi_d\Lambda(k), \xi=\pm1, 1\leq l\leq \#\Lambda'}|\sigma+k\cdot\omega+\xi\sqrt{\zeta_l}|>e^{-\frac 14 L^{\rho_2}}.
\end{align*}
From $\alpha\in{\rm DC}_{\nu}$ and   applying  Lemma \ref{Boulem} with $E=(\sigma+k\cdot\omega)^2$  gives the sublinear bound on the number of bad blocks (of size  $M\sim N_1^c$ for some $0<c\ll1$) contained in $\Pi_d\Lambda(k)$.    Next, applying  the resolvent identity of \cite{HSSY}  (cf. Lemma 4.2) yields for $B_k=R_{\Pi_d\Lambda(k)}({D}(\sigma)+\varepsilon\Delta)R_{\Pi_d\Lambda(k)},$  
\begin{align*}
\|B^{-1}_k\|&\leq e^{\frac 12 N_1^{\rho_2}},\\
|B_k^{-1}(n;n'')|&\leq e^{-\frac12|\log\varepsilon|\cdot|n-n''|}\ {\rm for}\ |n-n''|\geq {N_1}^{\rho_3}.
\end{align*}
Note that 
\begin{align*}
\tilde G_{\Lambda}(\sigma)=\bigoplus_{k\in\Pi_b\Lambda}B_k^{-1},
\end{align*}
where 
\begin{align*}
\tilde G_\Lambda(\sigma)=(R_\Lambda ({D}(\sigma)+\varepsilon\Delta)R_\Lambda)^{-1}.
\end{align*}
By taking account of all $k\in\Pi_b\Lambda,$ we obtain  that if 
\begin{align}
\nonumber& \min_{k\in\Pi_b\Lambda,\sqrt{N_1}\leq L\leq N_1,  \mathcal{ER}_d(L)\ni \Lambda'\subset \Pi_d\Lambda(k), \xi=\pm1,  1\leq l\leq \#\Lambda'}|\sigma+k\cdot\omega+\xi\sqrt{\zeta_l(k)}|\\
\label{2mel}&> e^{-\frac 14 L^{\rho_2}},
\end{align}
then 
\begin{align*}
\|\tilde G_{\Lambda}(\sigma)\|&\leq e^{\frac12N_1^{\rho_2}},\\
|\tilde G_{\Lambda}(\sigma)((k,n);(k'',n''))|&\leq \delta_{k,k'}e^{-\frac12|\log\varepsilon| \cdot|n-n''|}\ {\rm for}\ |n-n''|\geq N_1^{\rho_3}.
\end{align*}
Finally, since  $|n|>2N_1$,  we know that ${T}_{\phi}$ has the decay estimate
\begin{align*}
{T}_{\phi}((k,n);(k'',n''))\leq Ce^{-2\gamma N_1}.
\end{align*}
Using the Neumann series argument  gives that $\Lambda$  is $\sigma$-good  assuming \eqref{2mel} holds true.  This proves (1) of Lemma \ref{lsclem2}.  
\begin{rem}\label{remn1}
We would also like to remark that if we take into account of all these $\Lambda\in (k,n)+\mathcal{ER}_0(N_1)$ with $|(k,n)|\leq 10N$ and $|n|>2N_1$, then we get a sequence $\{\sqrt{\zeta_l}\}_{1\leq l\leq N^C}$. As a result,   there is a set $\Sigma_{N_1}\subset \R$ satisfying  ${\rm meas}(\Sigma_{N_1})\leq N^Ce^{-\frac14 N_1^{\rho_2/2}}\leq e^{-\frac15 N_1^{\rho_2/2}}$ so that,  for $\sigma\not\in\Sigma_{N_1}$, all $\Lambda\in(k,n)+\mathcal{ER}_0(N_1)$ with $|(k,n)|\leq 10N$ and $|n| >2N_1$ are $\sigma$-good. 
\end{rem}

\item[(2)] The conclusion is just a corollary of (1) by replacing  $N_1$ with $N_2$ in Remark~\ref{remn1}. 
\end{itemize}

This completes the proof of  Lemma \ref{lsclem2}.
\end{proof}
By combining Lemma \ref{lsclem1} and Lemma \ref{lsclem2},  we are ready to prove LDE at the scale $N.$ We first recall an important lemma. 

\begin{lem}[Cartan's estimate,  Proposition 14.1 in \cite{Bou05}]\label{mcl}
Let $T(\sigma)$ be a self-adjoint $N\times N$ matrix-valued function of a parameter $\sigma\in[-\xi, \xi]$  satisfying the following conditions:
\begin{itemize}
\item[(i)] $T(\sigma)$ is real analytic in $\sigma$ and has a holomorphic extension to
\begin{align*}
\mathbb{D}=\left\{z\in\mathbb{C}: \ |\Re z|\leq \xi,\ |\Im z|\leq \xi_1\right\}
\end{align*}
satisfying $\sup\limits_{z\in \mathbb{D}}\|T(z)\|\leq K_1,\  K_1\geq 1.$

\item[(ii)] For each $\sigma\in[-\xi, \xi]$, there is a subset $V\subset [1,N]$ with $\#V \leq M$ such that 
\begin{align*}
\|(R_{[1,N]\setminus V}T(\sigma)R_{[1,N]\setminus V})^{-1}\|\leq K_2, \ K_2\geq 1.
\end{align*}
\item[(iii)] Assume 
\begin{align*}
\mathrm{meas}\left(\{\sigma\in[-{\xi}, {\xi}]: \ \|T^{-1}(\sigma)\|\geq K_3\}\right)\leq 10^{-3}\xi_1(1+K_1)^{-1}(1+K_2)^{-1}.
\end{align*}
\end{itemize}
Let $0<\epsilon\leq (1+K_1+K_2)^{-10 M}.$ 
 Then we have
\begin{align}\label{mc5}
\mathrm{meas}\left(\left\{\sigma\in\left[-{\xi}/{2}, {\xi}/{2}\right]:\  \|T^{-1}(\sigma)\|\geq \frac1\epsilon\right\}\right)\leq Ce^{\frac{-c\log \frac1\epsilon}{M\log( M+K_1+K_2+K_3)}},
\end{align}
where $C, c>0$ are some absolute constants.
\end{lem}

\begin{proof}[Proof of Lemma \ref{lsclem}]

Recall that
\begin{align*}
N_2=N_1^{\widetilde C},\  N=N_2^{\widetilde C}.
\end{align*}
We assume that for $K=N_1, N_2$ and $\Lambda\in (0,n)+\mathcal{ER}_0(K)$ with $|n|\leq 2K$, $G_{\Lambda}(\sigma)$ satisfies the LDE.  

We first prove the sublinear bound on the number of  $\sigma$-bad $N_1$-elementary regions with centers in 
$[-N, N]^b\times( [-3N, 3N]^d \setminus [-2N_1, 2N_1]^d)$ for all $\sigma\in \R.$ For this purpose, an additional restriction on $\omega$ is required. Recalling (1) of Lemma \ref{lsclem2} and taking into account of all $\Lambda\in (k,n)+\mathcal{ER}_0(N_1)$ with $|(k,n)|\leq 10N$ and $|n|>2N_1$ yield   a sequence  $\{\sqrt{\zeta_l}\}_{\leq l\leq N^C}\subset \R$ with  $C=C(b,d)>0$ and each $\zeta_l$ depending only on $\alpha,\theta_0, m_0 $ (but not on $\sigma,\omega$), so that if $\Lambda=(k,n)+\mathcal{ER}_0(N_1)$ satisfies $|(k,n)|\leq 10N$,  $|n|>2N_1$ and 
$$\min_{k\in\Pi_b\Lambda, \xi=\pm1, 1\leq l\leq N^C}|\sigma+k\cdot\omega+\xi\sqrt{\zeta_l}|>e^{-\frac 14 N_1^{\rho_2/2}},$$
then $\Lambda$ is $\sigma$-good. So we can define the set
\begin{align}\label{tilomg}
\widetilde\Omega_N=\bigcap_{1\leq l,l'\leq N^C, \xi=\pm1, \xi'=\pm1, 0<|k|\leq 2N}\left\{\omega\in\Omega:\ |k\cdot\omega+\xi\sqrt{\zeta_l}-\xi'\sqrt{\zeta_{l'}}|>2 e^{-\frac 14 N_1^{\rho_2/2}}\right\}. 
\end{align}
Then we have ${\rm meas}(\Omega\setminus \widetilde\Omega_N)\leq e^{-\frac 15 N_1^{\rho_2/2}}.$ 
In the following we always assume $\omega\in \widetilde\Omega_N.$ Suppose that $\Lambda'\in(k',n')+\mathcal{ER}_0(N_1)$ satisfying  $|k'|\leq N$ and  $2N_1<|n'|\leq 10N$  is $\sigma$-bad. Then there are some $k_*\in \Pi_b\Lambda'\subset [-N,N]^b$, $\xi_*=\pm1$  and some $l_*\in[1, N^C]$ so that 
\begin{align*}
|\sigma+k_*\cdot\omega+\xi_*\sqrt{\zeta_{l_*}}|\leq e^{-\frac14 N_1^{\rho_2/2}}. 
\end{align*}
Now  let $\Lambda''\in(k'',n'')+\mathcal{ER}_0(N_1)$ with  $|k''|\leq N$ and  $2N_1<|n''|\leq 10N$   be another  $\sigma$-bad region. 
  From $\omega\in \widetilde\Omega_N,$ we  have 
\begin{align}
k_*\in \Pi_b\Lambda''.
\end{align}
In other words, we have established  all $\sigma$-bad elementary $N_1$-regions $\Lambda$  with centers $(k,n)$ satisfying $|k|\leq N$ and $2N_1<|n|\leq 10N$  must obey  
\begin{align}\label{krestr}
\Pi_b\Lambda\subset [k_*-2N_1, k_*+2N_1]^b. 
\end{align}
Next, we estimate the number of $\sigma$-bad elementary  $N_1$-regions satisfying \eqref{krestr}.  This needs to control the $n$-directions. We would like to apply  
Lemma~\ref{Boulem}  with $E_k=(\sigma+k\cdot\omega)^2$ and $|k-k_*|\leq 2N_1.$  Then there is $\Theta_k=\Theta_{N_1}(E_k)\subset \R$ with ${\rm meas}(\Theta_k)\leq e^{-N_1^{\rho_4}}$ so that if $\theta\notin \Theta_k$, then for all $Q\in \mathcal{ER}_{d,0}(N_1)$, we have 
\begin{align*}
\|\widetilde T^{-1}_{Q}(E_k;\theta)\|&\leq e^{\sqrt{N_1}},\\
|\widetilde T^{-1}_{Q}(E_k;\theta)(n;n')|&\leq e^{-\frac12|\log\varepsilon|\cdot|n-n'|}\ {\rm for}\ |n-n'|\geq {N_1}^{\rho_3},
\end{align*}
where $\widetilde T_Q(E_k;\theta)$ is given by \eqref{TQ} with $E=E_k.$
As done in the proof of Lemma~\ref{lsclem1}, we can also regard $\Theta_k$ as a semi-algebraic set of degree at most $N_1^C$. So the set of singular $\theta$ can be described as 
\begin{align*}
\bigcup_{k\in\Z^b,\ |k-k_*|\leq 2N_1}\Theta_k=\bigcup_{1\leq l\leq N_1^{C_1}} \tilde I_l,
\end{align*}
where each $\tilde I_l$ is an interval of length $|\tilde I_l|\leq e^{-\frac14 N_1^{\rho_4}}$.  Recall that $\rho_2>\frac23$. By the Neumann series argument and since $|{T}_\phi((k,n);(k',n'))|\leq Ce^{-2\gamma N_1}$ for $|n|>2N_1$, we get that  if $\Lambda \in (k,n)+\mathcal{ER}_0(N_1)$  is $\sigma$-bad satisfying  $|n|>2N_1$ and \eqref{krestr}, then 
\begin{align*}
\theta_0+n\cdot\alpha\in \bigcup_{1\leq l\leq N_1^{C_1}} \tilde I_l.
\end{align*}
Recall that $\alpha\in{\rm DC}_\nu$. This together with the pigeonhole principle implies the number of $\sigma$-bad $N_1$-regions $\Lambda\in (k,n)+\mathcal{ER}_0(N_1)$ satisfying \eqref{krestr} and $|n|>2N_1$ is at most $N_1^{2C_1}. $  Otherwise, there must be some $n\neq n'$ so that both $\theta_0+n\cdot\alpha\in\tilde I_l$ and $\theta_0+n'\cdot\alpha\in\tilde I_l$ for some $1\leq l\leq N_1^{C_1}$. Then  we have
\begin{align*}
N^{-C}\leq |(n-n')\cdot\alpha|\leq e^{-\frac12 N_1^{\rho_4}},
\end{align*}
a contradiction. 
In conclusion, assuming  $\omega\in\widetilde\Omega_N$,  we  have proven  that for all $\sigma\in\R$,   the number of $\sigma$-bad elementary $N_1$-regions $\Lambda$ with centers $(k,n)$ satisfying  $|k|\leq N$ and $2N_1<|n|\leq 10N$  is at most 
\begin{align*}
N_1^{2C_1}\leq N^{\rho_2/4}/2
\end{align*}
if $\widetilde C^2\geq C/\rho_2$ for some $C=C(b,d)>0.$  Combining with  Lemma~\ref{lsclem1},  we get  actually  the sublinear bound  $N^{\rho_2/4}$ of  all $\sigma$-bad elementary $N_1$-regions  with centers belonging to $\Lambda\subset  [-N,N]^b\times [-10N,10N]^{d}$.

Now let $\tilde N\in[\sqrt{N}, N]$ and $\Lambda(\tilde N)\in (0,n)+\mathcal{ER}_0(\tilde N)$ with $|n|\leq 3N$. We  apply  Lemma~\ref{mcl}
with  \begin{align*}
T(\sigma)={{H}}_{\Lambda(\tilde N)}(\sigma), \ \xi=\xi_1=e^{- 10N_1^{\rho_2}}.
\end{align*}
It remains to verify the assumptions of Lemma \ref{mcl}.
Obviously, one has
\begin{align*}
K_1=O(\tilde N).
\end{align*}
We let $V$ be the union of all $\sigma$-bad $\Lambda\in\mathcal{ER}(N_1)$ with centers belonging to $[-N, N]^b\times [-3N, 3N]^d$. Then by the above sublinear bound conclusion and appying the resolvent identity, we have 
\begin{align}\label{mb}
M=\#V\leq   N^{\rho_2/4},\ 
K_2=e^{2{N_1}^{\rho_2}}.
\end{align}
To apply the Cartan's  estimate, we need the scale $$N_2\sim N_1^{\widetilde C}.$$  Recall that the LDE hold  at the scale $N_2$. 
Applying the resolvent identity  and using  (2) of Lemma \ref{lsclem2} yield 
\begin{align*}
 \|T^{-1}(\sigma)\|\leq CN_2^Ce^{{N_2}^{\rho_2}}\leq e^{2{N_2}^{\rho_2}}=K_3 
\end{align*}
for $\sigma$ away from a set  of measure at most (since $\rho_1<\rho_2/2$)
$$C\tilde N^{C}e^{-{N_2^{\rho_1}}}+e^{-\frac15{N_2^{\rho_2/2}}}\leq e^{-{N_2^{\rho_1}}/{2}}.$$
It follows from  the assumption  $$\widetilde C\rho_1>\rho_2$$ 
that
\begin{equation*}
10^{-3}\xi_1(1+K_1)^{-1}(1+K_2)^{-1}> e^{-{N_2^{\rho_1}}/{2}}.
\end{equation*}
This verifies (iii) of Lemma \ref{mcl}.
If $\epsilon=e^{-\tilde N^{\rho_2}}$,  then one has   $\epsilon<(1+K_1+K_2)^{-10M}.$  Cover $[-10\tilde N,  10\tilde N]$ with disjoint intervals of length $e^{- 10N_1^{\rho_2}}$ and take into account of all $\Lambda(\tilde N)$ with $\tilde N\in [\sqrt{N}, N]$.
Then by \eqref{mc5} of Lemma \ref{mcl}, one obtains 
since $\rho_1<\rho_2/5$ that 
\begin{equation*}
\mathrm{meas}(\Sigma_{N})\leq CN^Ce^{ 10N_1^{\rho_2}}e^{-\frac{c{{\tilde N}^{\rho_2}}}{CN_1N_2  N^{\rho_2/4}\log  \tilde N}}\leq e^{-N^{\rho_2/5}}\leq  e^{-N^{\rho_1}}.
\end{equation*}
In addition, if $\sigma\notin\Sigma_N$, then  for all $\tilde N\in [\sqrt{N}, N]$,  $\Lambda\in (0,n)+\mathcal{ER}_0(\tilde N)$ with $|n|\leq 3 N$,
\begin{align}\label{sublg}
\|G_{\Lambda}(\sigma)\|&\leq e^{\tilde N^{\rho_2}}.
\end{align}

Finally,  for $\sigma\notin\Sigma_N$,  it suffices to prove the exponential off-diagonal decay of $G_{\Lambda}(\sigma)$ for $\Lambda\in (0,n)+\mathcal{ER}_0( N)$ with $|n|\leq 2 N$.  This will be completed using the coupling lemma (i.e., Lemma 4.2 of \cite{HSSY}).  
The sublinear bound can be established as follows: Assume that  $\mathcal{F}$ is a (any) family  of pairwise disjoint bad $M$-regions in  $\Lambda$ with $\sqrt{N} + 1 \leq  M\leq  2\sqrt{N} + 1$.   Then from Lemma~\ref{lsclem1}, 
$\#\mathcal{F}\leq N^{\rho_2/4}\leq \frac{N^{\rho_2}}{M}  $ since $\rho_2>\frac23.$
On the other hand, the sub-exponential growth of $\|G_\Lambda(\sigma)\|$ has been  given by \eqref{sublg}.  As a result, we get  the exponential off-diagonal decay estimate  with the decay rate $\gamma_N=\gamma_{N_2}-N^{-\kappa}. $

Collecting all the restrictions on $\widetilde C, \rho_1, \rho_2$ yields 
\begin{align*}
1>\rho_2>2/3,\, 0<\rho_1<\rho_2/5,\ \widetilde C^2\geq C/\rho_2, \ \widetilde C>\frac{\rho_2}{\rho_1}.
\end{align*}
This completes the  proof of Lemma~\ref{lsclem}.

\end{proof}
	{\begin{rem} We cannot use Lemma \ref{mcl} requiring  
 all $\Lambda\in(k,n)+\mathcal{ER}_0(N_1)$ with $|n|>2N_1$ are $\sigma$-good. This leads to removing  a set (of $\sigma$) of measure at most 
\begin{align*}
e^{-\frac23 N_1^{\rho_2}}<e^{-N^{\rho_1}}.
\end{align*} 
So we must have 
\begin{align*}
\frac 23N_1^{\rho_2}>N^{\rho_1}.
\end{align*}
To apply Lemma \ref{mcl}, one uses the sublinear bound on number of bad $\Lambda\in\mathcal{ER}(N_1)$
and applies  the resolvent identity. This will lead to the bound 
\begin{align*}
K_2=e^{\frac43 N_1^{\rho_2}}.
\end{align*}
To fulfill the condition (iii) of Lemma \ref{mcl}, we need to another scale 
\begin{align*}
N_1<N_2<N
\end{align*}
and use the resolvent identity to remove further a set (of $\sigma$) of measure at most 
\begin{align*}
e^{-\frac 12 N_2^{\rho_1}}.
\end{align*}
Recalling the measure bound of (iii) in Lemma \ref{mcl}, it requires that
\begin{align*}
10^{-3}\xi_1(1+K_1)^{-1}(1+K_2)^{-1}>e^{-\frac 12 N_2^{\rho_1}},
\end{align*}
that is 
\begin{align*}
{\frac 43 N_1^{\rho_2}}<\frac12 N_2^{\rho_1}.
\end{align*}
This is impossible since 
\begin{align*}
N_2^{\rho_1}<N^{\rho_1}<\frac 23N_1^{\rho_2}.
\end{align*}
\end{rem}}

\begin{proof}[Proof of Theorem~\ref{ldtthm}]
The proof is just a combination of Lemmas~\ref{inilem},~\ref{intlem} ~and~\ref{lsclem}. 
\end{proof}

\section{Nonlinear Analysis}
	Relying on the analysis in section~\ref{ldtsec},  we will prove the existence of Anderson localized states using a Newton scheme and the Lyapunov-Schmidt decomposition. 
	We iteratively solve the $Q$  and then the $P$-equations.
	The large deviation estimates in $\sigma$ in Theorem~\ref{ldtthm} will be converted into estimates in the amplitudes $a$ by using 
	a semi-algebraic projection lemma (cf. Lemma~\ref{proj} below), and amplitude-frequency modulation, $\omega=\omega(a)$.

	Throughout this section we use the following notation for convenience. Given $X,Y\geq0$,  we write $X\lesssim Y$ if there is some $C=C(b,d)>0$ so that $X\leq CY.$ We write  $X\sim Y$ if there is some $C=C(b,d)>1$ so that   $C^{-1}X\leq Y\leq C  X.$
	
	\subsection{The Newton scheme}
Assuming that the approximate solutions $q^{(l)}$ of \eqref{Peq} have been constructed  for $1\leq l\leq r$, we aim to construct $q^{(r+1)}$. Toward that purpose, we let 
$$ q^{(r+1)}=q^{(r)}+\Delta_{r+1} q.$$
Consider the linearized equation at $q^{(r)}$
\begin{align*}
 {{F}}( q^{(r)})+{{H}}\Delta_{r+1}  q=0,
\end{align*}
where 
\begin{align*}
H={{H}}(q^{(r)})&=D(0)+\varepsilon \Delta+\delta{T}_{ q^{(r)}}
\end{align*}
with ${T}_{ q^{(r)}}$ given by \eqref{toep} for $\tilde q=q^{(r)}$. 
{Due to the small divisors difficulty,  the invertibility of the infinitely  dimensional operator $H$ cannot be expected.  For this reason, we  would like to deal with the truncated finitely dimensional operator $H_\Lambda=R_\Lambda H R_\Lambda$  and establish  estimates on $H_\Lambda^{-1}$. The error of this truncation can be well controlled based on the exponential off-diagonal decay estimates of both $H$ and $H_\Lambda^{-1}$, which is compatible with the Newton iteration, cf. page 374 of  \cite{Bou98} or  chapter 18 of \cite{Bou05} for details.}  More precisely,  we may  instead solve the following smoothed (or truncated)  linearized equation 
\begin{align*}
 {{F}}( q^{(r)})+R_{[-M^{r+1}, M^{r+1}]^{b+d}\setminus\mathcal{S}}{{H}}R_{[-M^{r+1},  M^{r+1}]^{b+d}\setminus\mathcal{S}}\Delta_{r+1} q=0,
\end{align*}
which leads to estimating the Green's function 
\begin{align*}
{{G}}_{r+1}= \left(R_{[-M^{r+1}, M^{r+1}]^{b+d}\setminus\mathcal{S}}{{H}} R_{[-M^{r+1}, M^{r+1}]^{b+d}\setminus\mathcal{S}}\right)^{-1}. 
\end{align*}
The following estimates  suffice for the convergence of the approximate solutions with exponential decay: 
\begin{align*}
\| {{G}}_{r+1}\|&\leq  M^{(r+1)^C}, \\
| {{G}}_{r+1}((k',n');(k,n))|&\leq e^{-c|k-k'|-c |n-n'|}\ {\rm for} \ |k-k'|+|n-n'|\geq (r+1)^C
\end{align*}
for some $c, C>0.$ The main purpose of this section is  thus to establish the above estimates assuming $0<\varepsilon+\delta\ll1$ and additional restrictions on $\omega$ (and then $a$).  

	\subsection{Extraction of parameters}
		Note that  $u^{(0)}$ (cf. \eqref{u0}) can be represented in the symmetric form 
\begin{align}\label{lattq0}
u^{(0)}(t, n)=\sum_{(k,n)\in\mathcal{S}}q^{(0)}(k,n)\cos(e_l\cdot\omega^{(0)}t),
\end{align} 
where $
\mathcal{S}=\{(e_l,n^{(l)})\}_{l=1}^b\cup\{(-e_l,n^{(l)})\}_{l=1}^b
$ is given by Theorem \ref{mthm} and 
\begin{align}\label{q0defn}
q^{(0)}(e_l,n^{(l)})=q^{(0)}(-e_l,n^{(l)})=a_l/2\ {\rm for}\ 1\leq l\leq b.
\end{align}
The solutions $q$ on $\mathcal S$ are held {\it fixed}:  $q=q^{(0)}$ on $\mathcal S$, the 
$Q$-equations are used instead to solve for the frequencies. 
Solving  the equations at the initial step
\begin{align*}
{F}\big|_{\mathcal{S}}\left(q^{(0)}\right)=0
\end{align*}
leads to, for $1\leq l\leq b$, the initial modulated frequencies
\begin{align}
\nonumber\omega_l^{(1)}&=\sqrt{(\omega_l^{(0)})^2+C_{p+1}^{p/2}2^{-p}a_l^p\delta}\\
\label{omgdef}&=\omega_l^{(0)}+\frac{C_{p+1}^{p/2}2^{-p}a_l^p\delta}{\sqrt{(\omega_l^{(0)})^2+C_{p+1}^{p/2}2^{-p}a_l^p\delta}+\omega_l^{(0)}}
\end{align}
satisfying $$\det \left(\frac{\partial \omega^{(1)}}{\partial a}\right)\sim \delta^b.$$

Along the way, 
we will show that $\omega^{(r)}=\omega^{(0)}+O(\delta)$, for all $r=0, 1, 2, ...$. 
So $\omega^{(r)}\in \Omega$, the frequency set in Theorem~\ref{ldtthm}.
Hence Theorem~\ref{ldtthm} is at our disposal to solve the $P$-equations,
and moreover we can regard $\omega\in\Omega$
as  an independent parameter.

\subsection{The initial iteration steps}
	Let $M$ be a large integer. 
From \eqref{mk1thm} of Theorem \ref{clusthm}, we can set $L=(\varepsilon+\delta)^{-\frac{1}{10^4db^4}}$ and $\eta=(\varepsilon+\delta)^{\frac{1}{8b}}$. 
	Let $q^{(0)}(\omega,a)=q^{(0)}(a)$  with $(\omega,a)\in\Omega\times [1,2]^b=[-C\delta, C\delta]^b\times [1,2]^b$.
	
	We start by constructing $q^{(1)}(\omega, a)$. Obviously, 
	$${F}(q^{(0)})=O(\varepsilon+\delta)$$
	 and ${\rm supp}\  {F}(q^{(0)})\subset \Lambda_C$ for some $C=C(p, \mathcal{S})>0.$ It suffices to estimate
	\begin{align*}
	H_M^{-1}=(R_{\Lambda_M\setminus\mathcal{S}}({D}(0)+\varepsilon\Delta+\delta{T}_{q^{(0)}})R_{\Lambda_M\setminus\mathcal{S}})^{-1}.
	\end{align*}
	It follows from  \eqref{mk1thm} and the Neumann series argument  that 
	\begin{align*}
	\|H_M^{-1}\|&\leq 2(\delta+\varepsilon)^{-\frac{1}{4b}},\\
	|H_M^{-1}(j;j')|&\leq 2(\delta+\varepsilon)^{-\frac{1}{4b}} e^{-c|j-j'|}\  {\rm for}\ j\neq j'\in \Lambda_M\setminus\mathcal{S} 
	\end{align*}
	since $M\ll (\varepsilon+\delta)^{-\frac{1}{10^4db^4}}$ is fixed. We define 
	\begin{align*}
	\Delta_1q=-H_M^{-1}{F}(q^{(0)}),
	\end{align*}
	and consequently
	\begin{align*}
	q^{(1)}=q^{(0)}+\Delta_1q=q^{(0)}-H_M^{-1}{F}(q^{(0)}).
	\end{align*}



        Substituting $q^{(1)}$ into the $Q$-equations, we obtain
	\begin{align*}
	\omega_l^2&=(\omega_l^{(0)})^2+\varepsilon[\frac{2}{a_l}(\Delta q^{(1)})(\pm e_l,n^{(l)})](\omega,a)\\
	&\ \ +\delta[\frac{2}{a_l}(q^{(1)})_{*}^{p+1}(\pm e_l,n^{(l)})](\omega,a), \, \l= 1, 2, ..., b.
	\end{align*}
	Since the second and the third terms are smooth in $\omega$ and $a$,  using the. implicit function theorem yields $\omega^{(2)}=\omega^{(2)}(a)=\{\omega_l^{(2)}(a)\}_{l=1}^b$, written in the form,
	\begin{align*}
	\omega_l^{(2)}(a)=\omega_l^{(0)}+(\varepsilon+\delta) \varphi_l^{(2)}(a), \,  \l= 1, 2, ..., b,
	\end{align*}
	with a smooth function $\varphi_l^{(2)}$. (Compare with \eqref{omgdef}.)
	Define
	$$\Delta_rq=-H_{M^{r}}^{-1}{F}(q^{(r-1)}),$$
	and
	$$q^{(r)}=q^{(r-1)}+\Delta_rq.$$
	The above constructions can be performed for $r_0$ steps,  with $r_0$ satisfying
	\begin{align*}
	M^{r_0}\leq  (\varepsilon+\delta)^{-\frac{1}{10^4db^4}}<M^{r_0+1}.
	\end{align*}
	So we have obtained  $q^{(l)}=q^{(l)}(\omega,a)$ ($1\leq l\leq r_0$)  for all $(\omega ,a)\in \Omega\times [1,2]^b$. 
	
	\subsection{The Inductive Theorem}

	We introduce the Inductive Theorem.
	  \begin{thm}
	 For $r\geq r_0,$ we have
	\begin{itemize}
	\item[{\bf(Hi)}] ${\rm supp}\ q^{(r)}\subset \Lambda _{M^r}$. 
	\item[{\bf (Hii)}] $\|\Delta_r q\|<\delta_r,\ \|\partial \Delta_r q\|<{\bar\delta}_{r}$, where $\partial$ refers to derivation in $\omega$ or $a$,   $$q^{(r)}=q^{(r-1)}+\Delta_r q,$$
	 and $\|\cdot\|=\sup_{\omega,a}\|\cdot\|_{\ell^2(\Z_*^{b+d})}.$
	\begin{rem}
	The size of $\delta_r, \bar\delta_r$ will satisfy $\log\log \frac{1}{\delta_r+\bar\delta_r}\sim r.$ Precisely, we have 
	\begin{align*}
	\delta_r<\sqrt{\varepsilon+\delta}e^{-(\frac{4}{3})^r},\ \bar\delta_r<\sqrt{\varepsilon+\delta}e^{-\frac12(\frac{4}{3})^r}.
	\end{align*}
	\end{rem}
	\item[{\bf (Hiii)}] $|q^{(r)}(k,n)|\leq e^{-c(|k|+|n|)}$ for some  $c>0.$
	\begin{rem}
	The constant $c>0$ will decrease slightly along the iterations but remain bounded away from $0$. This will become clear in the proof. We also remark that $q^{(r)}$ can be defined as a $C^1$ function on the entire parameter space $(\omega,a)\in\Omega\times [1,2]^b$ by using a standard extension argument, cf. \cite{Bou98, BW08, KLW24}.  So one can apply the implicit function theorem to solve the $Q$-equations 
	leading to 
		\begin{align}\label{algeq}
	\omega_l^{(r)}(a)=\omega_l^{(0)}+(\varepsilon+\delta)\varphi_l^{(r)}(a) \ (1\leq l\leq b),
	\end{align}
	where   $\|\partial \varphi^{(r)}\|=\sup_{1\leq l\leq b}\|\partial \varphi_l^{(r)}\|\lesssim 1$ and $ \varphi^{(r)}=( \varphi^{(r)}_l)_{l=1}^b.$ By ({\bf Hii}), we have 
	\begin{align*}
	|\varphi^{(r)}-\varphi^{(r-1)}|\lesssim (\varepsilon+\delta)\|q^{(r)}-q^{(r-1)}\|\lesssim {(\varepsilon+\delta)}\delta_r.
	\end{align*}
	Denote by $\Gamma_r$ the graph of $\omega^{(r)}=\omega^{(r)}(a)$. We have $\|\Gamma_r-\Gamma_{r-1}\|\lesssim{(\varepsilon+\delta)}\delta_r.$ Recall that $\omega$ is given by \eqref{omgdef} and $\varphi^{(0)}=0$. Thus we  have a diffeomorphism from $\omega^{(r)}\in\Omega$ to $a\in[1,2]^b$. So we also write $a^{(r)}=a^{(r)}(\omega)$ for $\omega\in\Omega.$
	\end{rem}
	
	\item[{\bf (Hiv)}] There is  a collection $\mathcal{I}_r$ of intervals $I\subset \Omega\times [1,2]^b$ of size $M^{-r^C}$, so that 
	\begin{itemize}
	\item[{\bf(a)}] On each $I\in\mathcal{I}_r,$ $q^{(r)}(\omega,a)$ is given by a rational function in $(\omega,a)$
of degree at most $M^{r^3}$;
	\item[{\bf (b)}] For $(\omega,a)\in\bigcup_{I\in\mathcal{I}_r}I$, 
	\begin{align*}
	\|F(q^{(r)})\|\leq \kappa_r,\ \|\partial F(q^{(r)})\|\leq \bar\kappa_r,
	\end{align*}
	where $\partial$ refers to derivation in $\omega$ or $a$, and $\log\log\frac{1}{\kappa_r+\bar\kappa_r}\sim r$. More precisely, we have 
	\begin{align*}
	\kappa_r<\sqrt{\varepsilon+\delta}e^{-(\frac{4}{3})^{r+2}},\ \bar\kappa_r<\sqrt{\varepsilon+\delta}e^{-\frac12(\frac{4}{3})^{r+2}}.
	\end{align*}
	\item[{\bf (c)}] For $(\omega,a)\in\bigcup_{I\in\mathcal{I}_r}I$ and  $H=H({q^{(r-1)}})$, one has
	\begin{align}
	\label{hivc1}\|H_N^{-1}\|&\leq M^{r^C},\\
	\label{hivc2}|H_N^{-1}(j;j')|&\leq e^{-c|j-j'|}\  {\rm for}\ |j-j'|>r^C,	\end{align}
	where $j=(k,n), j'=(k',n')$ and $H_N$ refers to the restriction of $H$ on $[-N,N]^{b+d}\setminus \mathcal{S}$;
	\item[{\bf (d)}] Each $I\in\mathcal{I}_r$ is contained in some $I'\in\mathcal{I}_{r-1}$ and 
	\begin{align*}
	{\rm meas}\left(\Pi_a(\Gamma_{r-1}\cap(\bigcup_{I'\in\mathcal{I}_{r-1}}I'\setminus\bigcup_{I\in\mathcal{I}_r}I))\right)\leq M^{-\frac{r}{C(b)}},
	\end{align*}
	where $C(b)>0$ depends only on $b$, and  $\Pi_a$ denotes the projection of the set on the $a$-variables. 
	\end{itemize}
	\end{itemize}
	\end{thm}
	
	\subsection{Proof of the Inductive Theorem}
	
	We assume that the Inductive Theorem  holds  for $r_0\leq l\leq r.$ We will prove it for  $l=r+1.$
	  It suffices to establish ({\bf Hiv, c}) and ({\bf Hiv, d}),  and the other inductive  assumptions can be verified  as  in    \cite{BW08} and \cite{KLW24}. For this purpose, we set $N=M^{r+1}$ and 
	  $$N_1=(\log N)^{C'},$$
	where $C'>1$ will be specified below.  
	
	To establish ({\bf Hiv, c}), we first make approximations on $T_{q^{(l)}}$  for different $q^{(l)}$. It follows from 
	\begin{align}
	\|T_{q^{(r)}}-T_{q^{(r-1)}}\|\lesssim \|\Delta_r q\|<\delta_r,\ \log\log\frac{1}{\delta_r}\sim r,
	\end{align} 
	and assumptions on  $R_{M^r}H({q^{(r-1)}})R_{M^r}$ that  $R_{M^r}H({q^{(r)}})R_{M^r}$ also satisfies essentially  the estimates \eqref{hivc1} and \eqref{hivc2}.  For the set $$ W=([-N,N]^{b+d}\setminus [-M^r/2, M^r/2]^{b+d})\setminus{\mathcal{S}},$$ we use $N_1$-regions  to do estimates. We let
	$$r_1=\left[\frac{\log N_1}{\log \frac43}\right]+1.$$
	Then 
	\begin{align}\label{tq1}
	\|T_{q^{(r)}}-T_{q^{(r_1)}}\|\lesssim\delta_{r_1}\lesssim\sqrt{\varepsilon+\delta}e^{-(\frac43)^{r_1}}< e^{-N_1}.\end{align}
	So we can estimate  $H_{W}^{-1}(q^{(r)})$ with  $H_Q^{-1}(q^{(r_1)})$ for $Q\subset W$.  More precisely, we  want  to  show the following estimates hold for every $Q\in\mathcal{ER}(N_1)$ with $Q\subset W$:
	\begin{align}
\label{tqr11}\|H_{Q}^{-1}(q^{(r_1)})\|&\leq e^{N_1^{\rho_2}},\\
\label{tqr12}|H_{Q}^{-1}(q^{(r_1)})(j;j')|&\leq e^{-c|j-j'|}\  {\rm for}\ |j-j'|>N_1^{\rho_3}.	
\end{align}
This requires additional  restrictions on $(\omega, a).$ For this, we divide into the following cases:\\
	{\bf Case 1.}  Assume $Q\in (k,n)+\mathcal{ER}_0(N_1)$ with $|n|>2N_1$ and $Q\subset W.$  Denote by $\mathcal{C}_1$ the set of all these $Q$.   To estimate $H_Q^{-1}(q^{(r_1)})$, we can impose the following condition
	\begin{align*}
	\min_{\xi=\pm1, 1\leq l\leq N_1^C}|k\cdot\omega+\xi\sqrt{\zeta_l}|>e^{-\frac14N_1^{\rho_2/2}},
	\end{align*}
where $\zeta_l>1/2$ are given by  Lemma \ref{lsclem2}. For $k=0$, the above condition holds trivially, since $\zeta_l>1/2.$ For $k\neq 0$, we can remove $\omega$ directly.  
Assume that  $C'\rho_2>3.$ So by taking account of all $Q\in\mathcal{C}_1$ and by Fubini's theorem,  we can find $I_1\subset \Omega\times[1,2]^b$ with 
$${\rm meas}(I_1)\leq N^Ce^{-\frac14 (\log N)^{C'\rho_2}}\leq  e^{-(\log N)^2}\ll (\varepsilon +\delta)^bM^{-r/2},$$
 so that,  for all $(\omega, a)\notin I_1$ and all  $Q\in \mathcal{C}_1$,  $H_Q^{-1} (q^{(r_1)})$    satisfies  \eqref{tqr11} and \eqref{tqr12}. \\
  {\bf Case 2.} Assume $Q\in (k,n)+\mathcal{ER}_0(N_1)$ with $|n|\leq 2N_1$ and $Q\subset W.$  Denote by $\mathcal{C}_2$ the set of all these $Q$.  In this case it must be that $|k|\geq M^r/2.$ We will use the projection  lemma  of Bourgain   and LDT to remove $\omega.$ 
 For any $Q\in\mathcal{C}_2$, we can write $Q=(k,0)+Q_0 $ with $Q_0\in (0,n)+\mathcal{ER}_0(N_1)$ for $|n|\leq 2N_1,$ and $M^r/2\leq |k|\leq N.$  This motivates us to  consider 
 $$H(\sigma)=H(q^{(r_1)};\sigma)=D(\sigma)+\varepsilon \Delta+\delta T_{q^{(r_1)}},$$
which has been investigated in section \ref{ldtsec}. 
 Recall that $G_Q(\sigma)$ denotes the Green's function of $H(\sigma)$ restricted to $Q$. Then the  Toeplitz property in the $k$-direction of $H(\sigma)$ implies 
 $$H^{-1}_Q(q^{(r_1)})=G_Q(0)=G_{Q_0}(k\cdot\omega).$$
 Note that $\Omega_{N_1}$ given by Theorem \ref{ldtthm} is a semi-algebraic set of degree at most $N_1^C$ (for $N_1>(\varepsilon+\delta)^{-\frac{1}{10^4db^4}}$) and of measure at least $\delta^b-O(N_1^{-C})$.  We may assume $\Pi_{\omega}I\subset \Omega_{N_1}$ for each $I\in\mathcal{I}_{r_1}.$ 	
	So  we can apply  the LDE at the scale $N_1$ on each $I\in\mathcal{I}_{r_1}$:  For all $Q\in (0,n)+\mathcal{ER}_0(N_1)$ with $|n|\leq 2N_1,$ we have for $\sigma$ outside a set  of  measure at most $e^{-N_1^{\rho_1}}$, the following bounds,
		\begin{align}
\label{GQ1}\|G_{Q}(\sigma)\|&\leq e^{N_1^{\rho_2}},\\
\label{GQ2}|G_{Q}(\sigma)(j;j')|&\leq e^{-c|j-j'|}\  {\rm for}\ |j-j'|>N_1^{\rho_3}.	
\end{align}
For the usage of projection lemma of  Bourgain in the present setting (i.e., ${\rm meas}(\Omega)\sim \delta^b$), we need to make more precise descriptions of  the admitted  $\sigma$ first.  We say $\sigma $ is $Q$-bad if either \eqref{GQ1} or \eqref{GQ2} fails.  
\begin{lem}\label{sgmlem}
Let $C'\rho_1>3$. Fix  $Q\in (0,n)+\mathcal{ER}_0(N_1)$ with $|n|\leq 2N_1.$  Denote by $\Sigma_{Q}$ the set of  $Q$-bad $\sigma\in\R$. Then 
\begin{align*}
\Sigma_Q\subset \bigcup_{1\leq l\leq N_1^C}J_l,
\end{align*}
where each $J_l$ is an interval of length $\sim (\varepsilon +\delta).$
\end{lem}
\begin{proof}[Proof of Lemma \ref{sgmlem}]
Note that $N=M^{r+1}\geq M^{r_0+1}>(\varepsilon +\delta)^{-\frac{1}{10^5db^4}}$ from our assumptions. So we have  
\begin{align*}
N_1\sim (\log N)^{C'}>c\log^{C'} \frac{1}{\varepsilon+\delta}>\log^{2/\rho_1} \frac{1}{\varepsilon+\delta}
\end{align*}
assuming $C'\rho_1>3,$  which implies 
\begin{align}\label{edn1}
(\varepsilon+\delta)^{-1}<e^{N_1^{\rho_1/2}}.
\end{align}
So we can define for each $j=(k,n)\in Q$ and $\xi =\pm 1$ the set 
\begin{align*}
J_{j,\xi}=\{\sigma\in\R:\ |\sigma+k\cdot\omega+\xi \mu_n|\leq C (\varepsilon+\delta)\},
\end{align*}
where $C>1$ depends  only on $\|\Delta\|, \|T_{q^{(r_1)}}\|.$
From \eqref{edn1} and the Neumann series argument, we must have 
\begin{align*}
\Sigma_Q\subset \bigcup_{j\in Q,\xi=\pm1} J_{j,\xi}.
\end{align*}
This proves Lemma \ref{sgmlem}. 

\end{proof}

Now fix $I_0\in\mathcal{I}_{r_1}.$   Solving  the $Q$-equation at $r=r_1$ leads to the graph $\Gamma_{r_1}$. Then $\tilde I=\Pi_\omega(\Gamma_{r_1}\cap I_0)$ is an interval of size $\sim \varepsilon+\delta.$
For $\omega\in \tilde I$, let $\tilde \Sigma$ be the set of $\sigma\in\R$ so that  either \eqref{GQ1} or \eqref{GQ2} fails  for some $Q\in (0,n)+\mathcal{ER}_0(N_1)$ with $|n|\leq 2N_1$.  Then by Lemma \ref{sgmlem}, $\tilde \Sigma$ can be covered by $N_1^C$ intervals of size $\sim (\varepsilon+\delta).$
Pick $J$ to be one of such intervals  and consider $\mathcal{K}:=\tilde I\times (\tilde \Sigma\cap J)\subset I_1\times J\subset \R^b\times \R.$
We will show that $\mathcal{K}$ is a semi-algebraic set of degree $\deg \mathcal{K}\leq N_1^CM^{Cr_1^3}$ and measure ${\rm meas}(\mathcal{K})\leq C(\delta+\varepsilon)^b e^{-N_1^{\rho_1}}$. This measure bound follows directly from the LDT at scale $N_1$ and the Fubini's theorem. For the semi-algebraic description of $\mathcal{K}$, we need the following lemma. 
\begin{lem}[Tarski-Seidenberg Principle, \cite{Bou07}]\label{tsp}
Denote by $(x,y)\in\mathbb{R}^{d_1+d_2}$ the product variable. If ${S}\subset\mathbb{R}^{d_1+d_2}$ is semi-algebraic of degree $B$, then its projections $\Pi_x{S}\subset\mathbb{R}^{d_1}$ and
 $\Pi_y{S}\subset\mathbb{R}^{d_2}$ are semi-algebraic of degree at most $B^{C}$, where $C=C(d_1,d_2)>0$.
\end{lem}
We define $X\subset I_0\times J$ to be the set of all $(\omega,a,\sigma)$ so that either \eqref{GQ1} or \eqref{GQ2} fails  for some $Q\in (0,n)+\mathcal{ER}_0(N_1)$ with $|n|\leq 2N_1$. Then  we can regard $X$ as a semi-algebraic set of degree $N_1^C$.   From ({\bf Hiv})  and solving the $Q$-equation for $r=r_1$, we know that $\Gamma_{r_1}\cap I_0$ is given by an  algebraic equation  in $\omega ,a$ (cf. \eqref{algeq}) of degree $M^{Cr_1^3}$. 
Note also that 
\begin{align*}
\mathcal{K}=\Pi_{\omega, a}\left(X\cap ((\Gamma_{r_1}\cap I_0)\times \R)\right)
\end{align*}
which together with Lemma \ref{tsp} implies  $\deg \mathcal{K}\leq N_1^CM^{Cr_1^3}.$
 Our aim is to estimate 
 \begin{align}\label{hypest}
{\rm meas}_b(\bigcup_{M^r/2\leq |k|\leq N}\{\omega:\  (\omega, k\cdot\omega)\in \mathcal{K}\}).
 \end{align}
At this stage, we need a projection lemma  of  Bourgain.
\begin{lem}[\cite{Bou07}]\label{proj}
 Let ${S}\subset[0,1]^{d=d_1+d_2}$ be a semi-algebraic set of degree $\deg({S})=B$ and $\mathrm{meas}_d({S})\leq\eta$, where
$\log B\ll \log\frac{1}{\eta}.$
Denote by $(x, y)\in[0,1]^{d_1}\times[0,1]^{d_2}$ the product variable. Suppose
$ \eta^{\frac{1}{d}}\leq\epsilon.$
Then there is a decomposition of ${S}$ as
\begin{align*}
{S}={S}_1\cup{S}_2
\end{align*}
with the following properties: The projection of ${S}_1$ on $[0,1]^{d_1}$ has small measure
$$\mathrm{meas}_{d_1}(\Pi_{x}{S}_1)\leq {B}^{C(d)}\epsilon,$$
and ${S}_2$ has the transversality property
\begin{align*}
\mathrm{meas}_{d_2}(\mathcal{L}\cap {S}_2)\leq B^{C(d)}\epsilon^{-1}\eta^{\frac{1}{d}},
\end{align*}
where $\mathcal{L}$ is any $d_2$-dimensional hyperplane in $\R^d$ s.t.,
$\max\limits_{1\leq j\leq d_1}|\Pi_\mathcal{L}(\bm e_j)|<{\epsilon},$
where we denote by $\bm e_1,\cdots,{\bm e}_{d_1}$ the $ x$-coordinate vectors.
\end{lem}
\begin{rem} This Lemma admits scaling.
It is based on the Yomdin-Gromov triangulation theorem \cite{Gro87}; for a complete proof, see \cite{BN19}. 
\end{rem}
Taking $\eta=C(\varepsilon+\delta)^be^{-N_1^{\rho_1}}, \epsilon =M^r/10$,  
we obtain  by  $C'\rho_1>3$ that
$$C(\varepsilon+\delta)^be^{-\frac{1}{b+1} N_1^{\rho_1}}\leq C(\varepsilon+\delta)^be^{-\frac{1}{b+1}(\log N )^{C'\rho_1}}\ll \epsilon.$$
We have a decomposition 
	$$ \mathcal{K}=\mathcal{K}_1\cup\mathcal{K}_2,$$
	where ${\rm meas}\ (\Pi_\omega\mathcal{K}_1)\leq C(\varepsilon+\delta)^{b+1}N_1^CM^{Cr_1^3}M^{-r}\leq (\varepsilon+\delta)^{b+1}M^{-\frac r2}$ and the factor $(\varepsilon+\delta)^{2b+1}$ comes from scaling when applying Lemma \ref{proj}. 
	Note however that $|k|\geq M^r/2$ can not ensure $\min_{1\leq l\leq b}|k_l|\geq M^r/2$.  So the hyperplane $\{(\omega,k\cdot\omega)\}$  does not satisfy  the steepness condition. To address this issue,  we will  need to take into account of all possible directions  and apply Lemma \ref{proj} $b$ times, as done by Bourgain  (cf.  (3.26) of \cite{Bou07} and also \cite{LW24}). This then leads to an upper bound  $\
	C(\varepsilon+\delta)^{b+1}M^{-\frac{r+1}{C(b)}}$ on  \eqref{hypest}, where $C(b)>0$ only depends on $b$.    
	
	Taking into account of all $J$,  we have shown the existence of $\tilde I_1\subset \tilde I$ with ${\rm meas}_b(\tilde I_1)\leq C(\varepsilon+\delta)^{b+1}N_1^CM^{-\frac{r+1}{C(b)}}\leq C(\varepsilon+\delta)^{b+1}M^{-\frac{r+1}{C(b)}}$ so that for $\omega\in \tilde I\setminus \tilde I_1$,  the estimates \eqref{tqr11} and \eqref{tqr12} hold  for all  $Q\in\mathcal{C}_2$.  Let $I_0$ range over $\mathcal{I}_{r_1}$. The total measure removed from $\Pi_\omega\Gamma_{r_1}$ is at most $ C(\varepsilon+\delta)^{b+1}M^{r_1^C}M^{-\frac{r+1}{C(b)}}\leq (\varepsilon+\delta)^{b+1}M^{-\frac{r+1}{C(b)}} $. Since \eqref{tqr11} and \eqref{tqr12} allow $O(e^{-N_1})$ perturbation of $(\omega, a)$, and $\|\Gamma_r-\Gamma_{r_1}\|\lesssim \delta_{r_1}\ll e^{-N_1}$, we obtain a subset $\Gamma_r'\subset  \Gamma_r$ with ${\rm meas}(\Pi_\omega\Gamma_r')\leq  (\varepsilon+\delta)^{b+1}M^{-\frac{r+1}{C(b)}}$ so that \eqref{tqr11} and \eqref{tqr12} hold on 
	\begin{align*}
	\bigcup_{I\in\mathcal{I}_{r_1}}(I\cap (\Gamma_r\setminus\Gamma_r')),
	\end{align*}
	and hence on 
	\begin{align*}
	\bigcup_{I\in\mathcal{I}_{r}}(I\cap (\Gamma_r\setminus\Gamma_r')).
	\end{align*}
	
	Since  $M^{r^C}+e^{N_1^{\rho_2}}\ll  M^{(r+1)^C}$,  the  estimates \eqref{hivc1} and \eqref{tqr11} allow $O(M^{-(r+1)^C})$ perturbation of $(\omega, a)$. Combining conclusions in the above two cases and the  resolvent identity of \cite{HSSY} (cf. Lemma 3.10) gives a collection $\mathcal{I}_{r+1}$ of intervals of size $M^{-(r+1)^C}$
so that for $I\in\mathcal{I}_{r+1}$, 	
		\begin{align*}
\|H_{N}^{-1}(q^{(r)})\|&\leq M^{-(r+1)^C},\\
|H_{N}^{-1}(q^{(r)})(j;j')|&\leq e^{-c_{r+1}|j-j'|}\  {\rm for}\ |j-j'|>(r+1)^{C},
	\end{align*}
which concludes ({\bf Hiv, c}) at the scale $r+1.$ We remark that in the above exponential  off-diagonal decay estimates, $c_{r+1}\geq c_r-(\log M)^8(r+1)^{-8}$, which implies  $\inf_{r}c_r>0$. This explains why we need the sublinear distant off-diagonal decay (i.e., $\rho_3\in(0,1)$) in LDT.

Finally, we have 
\begin{align*}
{\rm meas}_b\left(\Pi_\omega(\Gamma_r\cap(\bigcup_{I'\in\mathcal{I}_{r}}I'\setminus\bigcup_{I\in\mathcal{I}_{r+1}}I))\right)\leq {\rm meas}(\Pi_\omega\Gamma_r')\leq (\varepsilon+\delta)^{b+1}M^{-\frac{r+1}{C(b)}}.
\end{align*}
Recalling that $\omega\to a$ is a $C^1$ diffeomorphism and $\det (\frac{\partial \omega}{\partial a})\sim \delta^{b}$, we obtain since $\varepsilon\lesssim \delta$ that 
\begin{align*}
{\rm meas}\left(\Pi_a(\Gamma_r\cap(\bigcup_{I'\in\mathcal{I}_{r}}I'\setminus\bigcup_{I\in\mathcal{I}_{r+1}}I)))\right)\leq M^{-\frac{r+1}{C(b)}},
\end{align*}
	which yields ({\bf Hiv, d}) at the scale $r+1.$
	\hfill $\square$
	
	
	
	\begin{proof}[Proof of Theorem \ref{mthm}]
	The proof of Theorem \ref{mthm} is a direct corollary of the inductive theorem. We refer to \cite{BW08, KLW24} for further details. 
	\end{proof}
		\section*{Acknowledgements}

 This work was supported by the National Key R\&D Program (2021YFA1001600).   Y. Shi was  supported by the NSFC  (12522110)  and 
W.-M. Wang acknowledges support from the
CY Initiative of Excellence, ``Investissements d'Avenir" Grant No. ANR-16-IDEX-0008. The authors would like to thank the reviewer for valuable suggestions. 
\section*{Declaration of competing interest}
On behalf of all authors, the corresponding author states that there is no conflict of interest.
\section*{Data availability}
Data sharing not applicable to this article as no datasets were generated or analysed during the current
study. 
\bibliographystyle{alpha}

\end{document}